\documentclass{article}

\usepackage[square,authoryear]{natbib}
\usepackage{marsden_article}
\usepackage{amsmath, amsfonts,amssymb,euscript,latexsym}

\usepackage[margin=1cm, font=small,format=plain,labelfont=bf,up,textfont=up]{caption}

\newcommand{\Osc}{\mathrm{Osc}}
\newcommand{\osc}{\mathrm{osc}}

\title{The Dynamics of a Rigid Body \\ in Potential Flow with Circulation}
\author{Joris Vankerschaver$^{a,b}$,
Eva Kanso$^c$ \& 
Jerrold E. Marsden$^a$ \\ \mbox{} \\
{\small $^a$ Control and Dynamical Systems} \\
{\small California Institute of Technology MC 107-81,} 
{\small Pasadena, CA 91125} \\ 
{\small $^b$ Department of Mathematical Physics and Astronomy} \\
{\small Ghent University, Krijgslaan 281, B-9000 Ghent, Belgium} \\
{\small $^c$ Aerospace and Mechanical Engineering} \\
{\small University of Southern California,}
{\small Los Angeles, CA 90089} \\ \mbox{} \\
 {\small E-mail: jv@caltech.edu, kanso@usc.edu, jmarsden@caltech.edu}
}

\begin{document}
\maketitle

\begin{abstract}
	We consider the motion of a two-dimensional body of arbitrary shape in a planar irrotational, incompressible fluid with a given amount of circulation around the body.  We derive the equations of motion for this system by performing symplectic reduction with respect to the group of volume-preserving diffeomorphisms and obtain the relevant Poisson structures after a further Poisson reduction with respect to the group of translations and rotations.  In this way, we recover the equations of motion given for this system by Chaplygin and Lamb, and we give a geometric interpretation for the Kutta-Zhukowski force as a curvature-related effect.  In addition, we show that the motion of a rigid body with circulation can be understood as a geodesic flow on a central extension of the special Euclidian group $SE(2)$, and we relate the cocycle in the description of this central extension to a certain curvature tensor.
\end{abstract}




\tableofcontents

\section{Introduction}




We consider the motion of a rigid planar body, whose shape is not necessarily circular, immersed in a two-dimensional perfect fluid.  We assume that the vorticity of the fluid vanishes, but we allow for a non-zero amount of circulation around the rigid body.   This dynamical system is of fundamental importance in aerodynamics and was studied, among others, by Chaplygin, Kutta, Lamb, and Zhukowski; see \cite{Ko1993}, \cite{Ko2003}, and \cite{BoMa06} for a detailed overview of the literature and background information on this subject. A consequence of the non-zero circulation is the presence of a gyroscopic \emph{lift force} acting on the rigid body, which is proportional to the circulation.  This force is referred to as the \emph{Kutta-Zhukowski force}.

While various aspects of this dynamical system have been discussed throughout the fluid-dynamical literature (see for instance \cite{MiTh1968} or \cite{Batchelor} for a modern account), deeper insight into the Hamiltonian structure of this system had to wait until the work of \cite{Ko1993} and \cite{BoMa06}.  These authors identified a Hamiltonian structure for the equations of motion found by Chaplygin and used this structure to shed further light onto the integrability of the system and to investigate chaoticity.  This pioneering work raised several questions worth investigating in their own right. The first and foremost is that the non-canonical Hamiltonian structure of the rigid body with circulation is obtained \emph{by inspection}, which immediately begs the question whether there are other, more fundamental reasons as to why this should be so.   In this paper, we address this question following the approach of \cite{Ar66}. In particular, we model the dynamics of the rigid body moving in a perfect fluid as a geodesic flow on an infinite-dimensional manifold.   On this space, several symmetry groups act, and by performing successive reductions with respect to each of these groups, we obtain the Hamiltonian structure and equations of motion for the rigid body with circulation in a geometric way.  This approach is also used in \cite{VaKaMa2009} to derive the equations of motion for a rigid body interacting with point vortices.



In this geometric approach, the symmetry groups of the fluid-solid problem are the \emph{\textbf{particle relabeling symmetry}} group (the group of volume-preserving diffeomorphisms of the fluid reference space), and  the \emph{\textbf{special Euclidian group}} $SE(2)$ of uniform solid-fluid translations and rotations.   Symplectic reduction with respect to the former group is equivalent to the specification of the vorticity field of the fluid: in the problem considered here, this amounts to requiring that the external vorticity vanishes and that there is a given amount of circulation around the body.  The group $SE(2)$ acts on the system by combined solid-fluid transformations.  After Poisson reduction with respect to this group, we obtain the Hamiltonian structure of \cite{BoMa06}.  

One straightforward advantage of the geometric description is that it provides immediate insight into the dynamics, which might be harder to obtain by non-geometric means.  As an illustration, we show that the symplectic leaves for the Poisson structure obtained by \cite{BoMa06} are orbits of a certain \emph{\textbf{affine action}} of $SE(2)$ onto $\mathfrak{se}(2)^\ast$, which we compute explicitly.  We show that these leaves are paraboloids of revolution, hence providing a geometric interpretation of a result of Chaplygin.

An additional result of using geometric reduction is that we obtain a new interpretation for the classical 
Kutta-Zhukowski force.  In particular, we show that the Kutta-Zhukowski force is proportional to the curvature of a natural fluid-dynamical connection, called the Neumann connection, which encodes the influence of the rigid body on the surrounding fluid.  In this way, we exhibit interesting parallels between the dynamics of the rigid body with circulation and that of a charged particle in a magnetic field: both systems are acted upon by a gyroscopic force, the Kutta-Zhukowski force and the Lorentz force, respectively.  The reduction procedure to recover the Hamiltonian description of \cite{BoMa06} then turns out to be similar to the way in which \cite{Sternberg1977} derives the equations for a charged particle by including the magnetic field directly into the Poisson structure.

Lastly, we also establish an analogue between the \emph{\textbf{Kaluza-Klein construction}} for magnetic particles and the dynamics of the rigid body with circulation.  In the conventional Kaluza-Klein description, the dynamics of the magnetic particle is made into a geodesic motion by extending the configuration space and including the magnetic potential into the metric.  Given the similarity between the Lorentz and the Kutta-Zhukowski force, a natural question is then whether a similar geometric description exists for the rigid body with circulation.  It turns out that this is indeed so: the extended configuration space in this case is a central extension of $SE(2)$ by means of natural cocycle, which is again related to the curvature giving rise to the Kutta-Zhukowksi force.  This extension of $SE(2)$ is known as the \emph{\textbf{oscillator group}} in quantum mechanics; see \cite{Streater1967}.  The rigid body with circulation is then described by geodesic motion on the oscillator group.

Our description of the rigid body with circulation as a geodesic flow on a central extension is similar to the work of \cite{OvKh1987}, who showed that the Korteweg-de Vries equation can be interpreted as a geodesic motion on the Virasoro group, an extension of the diffeomorphism group of the circle.  It is also worth stressing that the similarity with the classical Kaluza-Klein picture cannot be taken too far: in the Kaluza-Klein description, one modifies the metric to take into account the magnetic field, whereas we leave the metric on $SE(2)$ essentially unchanged and instead we deform the multiplication structure by means of a cocycle, giving rise to a central extension of $SE(2)$.

\paragraph{Outline of this Paper.}  We begin the paper by giving an overview of the classical literature on fluid-structure interactions in section~\ref{sec:prel}.  In section~\ref{sec:fsint} we recall some of the geometric concepts that arise in this context, most notably the particle relabeling symmetry group and the Neumann connection, and we give a brief summary of cotangent bundle reduction, which we use in section~\ref{sec:red} to derive the Chaplygin-Lamb equations describing a rigid body with circulation.  The geometry of the oscillator group and the link with the rigid body with circulation are explained in section~\ref{sec:geodosc}, and we finish the paper in section~\ref{sec:outlook} with a brief discussion of future work.   In appendix~\ref{appendix:rigidgroup} and \ref{appendix:diffgroup}, some elementary results are collected about the geometry of the special Euclidian group $SE(2)$ and the group $\mathrm{Diff}_{\mathrm{vol}}$ of volume-preserving diffeomorphisms.

\paragraph{Acknowledgements.}  We would like to thank Scott Kelly, Jair Koiller, Tudor Ratiu and Banavara Shashikanth for useful suggestions and interesting discussions.

J. Vankerschaver is supported through a postdoctoral fellowship from the Research Foundation
-- Flanders (FWO-Vlaanderen).  Additional financial
support from the Fonds Professor Wuytack is gratefully acknowledged.
E. Kanso and J. E. Marsden would like to acknowledge the support of the National Science Foundation through the grants CMMI 07-57092 and  CMMI 07-57106, respectively.

\section{Body-Fluid Interactions: Classical Formulation} 
\label{sec:prel}

Consider a planar body moving in an infinitely large volume of an incompressible and
inviscid fluid $\mathcal{F}$ at rest at infinity. The body $\mathcal{B}$
is assumed to occupy a simply connected region whose  boundary 
can be conformally mapped to a unit circle, and it is considered to be uniform
and neutrally-buoyant (the body weight is balanced by the force of buoyancy).
Introduce an orthonormal inertial 
frame $\{{\bf e}_{1,2,3} \} $ where $\{{\bf e}_1,{\bf e}_2\}$ span the plane
of motion and ${\bf e}_3$ is the unit normal to this plane. 
The configuration of the submerged 
rigid body can then be described by a rotation $\theta$ about 
$\mathbf{e}_3$ and a translation $\mathbf{x}_o =
x_o\mathbf{e}_1 + y_o\mathbf{e}_2$ 
of a point $O$ (often chosen to coincide with the conformal center of the body). 
The angular and translational velocities expressed relative to the inertial
frame are of the form
$ \dot{\theta}\,\mathbf{e}_3$ and $\mathbf{v} = v_x \,\mathbf{e}_1 + v_y\, \mathbf{e}_2$ 
where $v_x = \dot{x}_o$, $v_y = \dot{y}_o$ (the dot denotes derivative with respect to time $t$).
It is convenient for the following development to 
introduce a moving  frame $\{{\bf b}_{1,2,3} \} $ attached to the body.
The point transformation from the body to the inertial frame can be represented as
\begin{equation}\label{eq:rigidmotion}
\mathbf{x} = R_{\theta} \mathbf{X} + \mathbf{x}_o, \qquad R_{\theta} =  
\begin{pmatrix}
 \cos\theta & -\sin\theta \\ \sin\theta & \cos\theta 
 \end{pmatrix},
\end{equation}

where $\mathbf{x} = x\,\mathbf{e}_1 + y\, \mathbf{e}_2$ and 
$\mathbf{X} = X\, \mathbf{b}_1 + Y\, \mathbf{b}_2,$  while vectors
transform as $\mathbf{v} = R_{\theta} \mathbf{V}.$ The angular and translational 
velocities expressed in the body frame take the form
$\boldsymbol{\Omega} = \Omega \, \mathbf{b}_3$ (where  $\Omega = \dot{\theta}$) 
and $\mathbf{V}=V_x\mathbf{b}_1 + V_y \mathbf{b}_2$ (where
 $V_1 = \dot{x}_o\cos{\theta} + \dot{y}_o\sin\theta$ 
and $V_2 = -\dot{x}_o\sin{\theta} + \dot{y}_o\cos\theta$).  
Note that the orientation
and position $(\theta,x_o,y_o)$ form an element of $SE(2)$, the group
of rigid body motions in $\mathbb{R}^2$. The velocity in the body-frame
 $\zeta = (\Omega, V_{x},V_{y})^T$, where $()^T$ denotes the transpose operation, 
 is an element of the 
vector space $\mathfrak{se}(2)$ which is the space of infinitesimal rotations and translations
in $\mathbb{R}^2$ and is referred to as the Lie algebra of $SE(2)$; for more details on the rigid
body group and its Lie algebra, see Appendix~\ref{appendix:rigidgroup} and references therein.

\paragraph{Fluid Motion.} Let the fluid fill the complement of the body in $\mathbb{R}^2$. 
The reference configuration of the fluid will be denoted by $\mathcal{F}_0$, 
and that of the body by $\mathcal{B}_0$.  The space taken by the fluid at a generic
time $t$ will be denoted by $\mathcal{F}$.  Note however that as time
progresses, the position of the body changes and hence so does its
complement $\mathcal{F}$.   In geometric mechanics, the mapping from
the reference configuration $\mathcal{F}_0$ to the fluid domain $\mathcal{F}$ can be
expressed as an element of the group of volume-preserving diffeomorphisms
reviewed in Appendix~\ref{appendix:diffgroup}.  In Section~\ref{sec:fsint}, 
we present an extension of this geometric approach to the solid-fluid 
problem considered here but beforehand, we briefly describe, 
using the classical vector calculus approach, the fluid motion 
and the equations governing the motion of the submerged body.

The fluid velocity $\mathbf{u}$ can be written using the 
Helmholtz-Hodge decomposition  as follows
\begin{equation}\label{eq:u}
\mathbf{u} = \nabla \Phi_\zeta  \ + \ \mathbf{u}_{\rm v}  ,
\end{equation}
and has to satisfy the impermeability boundary condition on the boundary of the solid body. Now we explain the two terms in this decomposition.

The potential function $\Phi_\zeta$ is harmonic 
and represents the irrotational motion of the fluid generated by
motion of the body.  
The subscript $\zeta$ refers to the fact that $\Phi_\zeta$ is determined by the body velocity $\zeta$.

We emphasize that the motion
of the body inside the fluid does not generate vorticity and only causes
irrotational motions the fluid. The potential function $\Phi_\zeta$ is a
solution to Laplace's equation $\Delta \Phi_\zeta = 0,$
subject to the boundary conditions 
\begin{equation} \label{Neumann}
\Delta \Phi_\zeta = 0, \qquad 
\left. \frac{\partial \Phi_\zeta}{\partial n} \right|_{\partial \mathcal{B}} =
(\boldsymbol{\Omega} \times \mathbf{X} + \mathbf{V}) \cdot \mathbf{n}, \qquad 
\left. \nabla \Phi_\zeta \right|_{\infty} = 0.
\end{equation}
By linearity of Laplace's equation, one can write, following Kirchhoff (see~\cite{lamb}),
\begin{equation}
\label{eq:velpot}
\Phi_\zeta =  \Omega \Phi_{\Omega} + V_x \Phi_x + V_y \Phi_y,
\end{equation}
where
$\Phi_\Omega, \Phi_x,\Phi_y$ are called velocity potentials
and are solutions to Laplace's equation subject to
the boundary conditions on $\partial \mathcal{B}$
\begin{equation}\label{eq:neumann1}
\begin{split}
\left.\dfrac{\partial\Phi_\Omega}{\partial n}\right|_{\partial \mathcal{B}}  \ = \
(\mathbf{X} \times  \mathbf{n} )\cdot \mathbf{b}_3  \ ,
\quad
\left.\dfrac{\partial\Phi_x}{\partial n} \right|_{\partial \mathcal{B}}   \ = \ \mathbf{n} \cdot \mathbf{b}_{1}  \ ,
\quad  \ \
\left.\dfrac{\partial\Phi_y}{\partial n} \right|_{\partial \mathcal{B}}   \ = \ \mathbf{n} \cdot\mathbf{b}_{2}  .
\end{split}
\end{equation}

The velocity  $\mathbf{u}_{\rm v}$  
is a divergence-free vector field and can be written
as $\mathbf{u}_{\rm v}= \nabla \times \boldsymbol{\Psi} + \mathbf{u}_\Gamma$,
where $\nabla \times \boldsymbol{\Psi}$ describes the fluid velocity due to ambient
vorticity and $\mathbf{u}_\Gamma$ describes the fluid velocity due to a net circulatory flow 
around the submerged body.
The vector potential  $\boldsymbol{\Psi}$ 
satisfies $\Delta \boldsymbol{\Psi} = -\boldsymbol{\omega}$, where $\boldsymbol{\omega} 
= \nabla \times \mathbf{u}_{\rm v}$ is the vorticity field, 
subject to the boundary conditions $(\nabla \times  \boldsymbol{\Psi}) \cdot  
\mathbf{n}=0$ on $\partial \mathcal{B}$ and $\nabla \times  \boldsymbol{\Psi}=0$ at infinity. 
Clearly, in the absence of ambient vorticity, $\boldsymbol{\Psi}$ is harmonic
and can be written for planar flows as $\boldsymbol{\Psi}= \Psi \mathbf{e}_3$, 
where $\Psi$ is  referred to as the stream function and satisfies the boundary conditions 
(see~\cite[\S9.40]{MiTh1968})
\begin{equation} \label{boundary}
\left. 	\Psi \right|_{\partial \mathcal{B}} = V_x Y - V_y X - \frac{\Omega}{2}(X^2 + Y^2) = 
\ \text{function of time only}. 
\end{equation}

The harmonic vector field $ \mathbf{u}_\Gamma$ is non-zero only when
there is a net circulatory flow around the body; it satisfies $\nabla
\cdot {\bf u}_{\Gamma} = 0$ and $\nabla \times {\bf u}_{\Gamma}= 0$
(i.e., $\Delta \mathbf{u}_\Gamma=0$) and the boundary conditions $
{\bf u}_{\Gamma} \cdot {\bf n} = 0$ on $\partial \mathcal{B}$ and
${\bf u}_{\Gamma} = 0$ at infinity. Note that, in three dimensional
flows, one does not need the harmonic vector field
$\mathbf{u}_\Gamma$.\footnote{ In three dimensions, any closed curve
  in the exterior of a {\it bounded} body is contractible, so the
  harmonic vector field $\mathbf{u}_\Gamma$ may be set to zero. This
  result is due to the {\em Poincar\'{e} Lemma} which can be alternatively
  stated as follows: a closed one-form on a 
  (sub)-manifold with trivial first cohomology is globally exact.
 }
 
 A harmonic stream function $\Psi_\Gamma$ associated with the circulation around the planar body
can be defined such that $\left. \Psi_\Gamma\right|_{\partial \mathcal{B}}=$ constant. 
The function $\Psi_\Gamma$ can be found using a conformal transformation 
that relates the flow field in the region
exterior to the body to that in the region exterior to the unit circle. For concreteness,
let $(\tilde{X},\tilde{Y})$ denote the body coordinates in the circle plane (which are measured relative
to a frame attached to the center of the circle as a result of choosing
the origin of the body frame in the physical plane to be placed at the conformal center).
The stream function $\Psi_\Gamma$ can be readily obtained by observing that 
the effect of having a net circulation $\Gamma$ around the body is equivalent to placing a point
vortex of strength $\Gamma$ at the center of mass of the body; namely,
\begin{equation} \label{circulation}
  \mathbf{u}_\Gamma = \nabla \times (\Psi_\Gamma \,  \mathbf{e}_3),  \qquad
  \Psi_\Gamma =   \frac{\Gamma}{4\pi} \log (\tilde{X}^2 + \tilde{Y}^2).
\end{equation}
Alternatively, the circulatory flow could be obtained as the gradient of a harmonic potential $\Phi_\Gamma$
(the harmonic conjugate to $\Psi_\Gamma$ satisfying the Cauchy-Riemann relations). Note that
$\Phi_\Gamma$ would have to satisfy $\left. \nabla \Phi_\Gamma \cdot \mathbf{n} \right|_{\partial \mathcal{B}} = 0$. 

\paragraph{Kinetic Energy.} The kinetic energy of the fluid-solid system is simply the 
sum of the kinetic energies for both constitutive systems:
\begin{equation} \label{Tkin}
	T =  T_{\mathrm{fluid}} + T_{\mathrm{body}} = 
	\frac{1}{2}\int_{\mathcal{F}} \left\Vert \mathbf{u} \right\Vert^2 dV + 
	\frac{1}{2} m \left\Vert \mathbf{V}\right\Vert^2 + \frac{1}{2} \mathbb{I} \Omega^2.
\end{equation}
where $dV$ is a standard volume (more precisely, area) element  on $\mathbb{R}^2$ . The kinetic energy
of the rigid body can be readily rewritten in the form
\begin{equation}\label{eq:Tbody}
T_{\mathrm{body}} = \dfrac{1}{2} \zeta^T \mathbb{M}_{b} \zeta, \qquad 
\mathbb{M}_b := \begin{pmatrix}
      \mathbb{I} & 0 \\
      0 & m \mathbf{I} 
    \end{pmatrix}.
\end{equation}
where $\mathbf{I}$ is the $2$-by-$2$ identity matrix. The kinetic energy of the fluid
can be written as
\begin{equation} \label{eq:T}
\begin{split}
T_{\mathrm{fluid}}  & = \dfrac{1}{2}\int_{\mathcal{F}} \left\Vert \mathbf{u} \right\Vert^2 dV = 
\dfrac{1}{2}\int_{\mathcal{F}} \nabla(\Phi_\zeta + \Phi_\Gamma)\cdot  \nabla(\Phi_\zeta + \Phi_\Gamma) \, dV \\
& = \dfrac{1}{2}\int_{\mathcal{F}} \nabla\Phi_\zeta \cdot  \nabla\Phi_\zeta  \, dV
+ \int_{\mathcal{F}} \nabla\Phi_\zeta \cdot \nabla \Phi_\Gamma \,  dV
+ \dfrac{1}{2}\int_{\mathcal{F}} \nabla \Phi_\Gamma \cdot  \nabla \Phi_\Gamma \, dV . \\ 
\end{split}
\end{equation}
The first term in  \eqref{eq:T}  can be rewritten using the divergence theorem, then employing~\eqref{eq:velpot}
and~\eqref{eq:neumann1}, as follows
\begin{equation} \label{eq:Tf_zeta}
\begin{split}
\dfrac{1}{2}\int_{\mathcal{F}} \nabla\Phi_\zeta \cdot  \nabla\Phi_\zeta  \, dV = 
 \dfrac{1}{2} \oint_{\partial \mathcal{B}} \Phi_\zeta \dfrac{\partial \Phi_\zeta}{\partial n} \, dS = 
 \dfrac{1}{2} \zeta^T \mathbb{M}_{f} \zeta, 
\end{split}
\end{equation}
where $ \mathbb{M}_{f} $ is a $3\times3$ added mass matrix.  Now, using the fact that $\nabla \Phi_\zeta$ and $\nabla \Phi_\Gamma$ are $L_2$-orthogonal: since $\nabla \Phi_\Gamma = \nabla \times (\Psi_\Gamma \,  \mathbf{e}_3)$, where $\Psi_\Gamma$ is uni-valued, we have 
\begin{align*}
	\int_{\mathcal{F}} \nabla \Phi_\zeta \cdot  \nabla \Phi_\Gamma \, dV & = 
	\int_{\mathcal{F}} \nabla \Phi_\zeta \cdot \nabla \times (\Psi_\Gamma \,  \mathbf{e}_3) \, dV 
	= \oint_{\partial \mathcal{B}} \Phi_\zeta \frac{\partial \Psi_\Gamma}{\partial s} \, dS = 0,
\end{align*}
since $\Psi_\Gamma$ is constant on the boundary $\mathcal{B}$, so that the tangential derivative $\partial \Psi_\Gamma/ \partial s$ vanishes.  The last term in \eqref{eq:T} can be treated as follows:  in polar coordinates, $\Phi_\Gamma(r, \theta) = \Gamma \theta$, so that if we enclose the fluid-solid system in a large circular box of radius $\Lambda$, 
\[
	\int_{\mathcal{F}} \nabla \Phi_\Gamma \cdot  \nabla \Phi_\Gamma \, dV = 2\pi \Gamma^2 \int_R^\Lambda \frac{dr}{r} = 2\pi \Gamma^2 \log \frac{\Lambda}{R}.
\]
This constant term diverges logarithmically as $\Lambda \rightarrow +\infty$.  To remedy this, we regularize the kinetic energy by discarding this infinite contribution -- a remedy also used in~\cite{lamb,BoMa06}. That is, we consider the kinetic energy of the solid-fluid system to be given by
\begin{equation}
\label{eq:T_solidfluid}
T =  \dfrac{1}{2} \zeta^T (\mathbb{M}_b + \mathbb{M}_f) \zeta .
\end{equation}


\paragraph{Equations of Motions.} 
The equations governing the motion of the body in potential flow with non-zero circulation around the body
but in the absence of ambient vorticity are of the form
\begin{equation}
\begin{split} \label{EoM1}
	\dot{\Pi} & = (\mathbf{P} \times \mathbf{V})\cdot \mathbf{b}_3  \\
	\dot{\mathbf{P}} & =  \mathbf{P} \times \boldsymbol{\Omega} + \Gamma \mathbf{b}_3 \times \mathbf{V}
\end{split}
\end{equation}
Here, $\Pi$ and $\mathbf{P}$ denote the angular and linear momenta of the solid-fluid system
expressed in the body frame. They are given in terms of the velocity in body frame by (here, $T$ is given by~\eqref{eq:T_solidfluid})
\begin{equation}
\Pi = \dfrac{\partial T}{\partial \Omega}, \qquad 
\mathbf{P}  = \dfrac{\partial T}{\partial \mathbf{V}}. 
\end{equation}

One of the main objectives of this paper is to use the methods of
geometric mechanics (particularly the \textit{reduction by stages} approach)
to derive equations~\eqref{EoM1} governing the motion of
the body in potential flow and with non-zero circulation. 
The case of a body of arbitrary geometry interacting with external
point vortices is addressed in~\cite{VaKaMa2009}.  It would be of interest to extend these results to the case of a rigid body with circulation moving in the field of point vortices.


\section{Body-Fluid Interactions:  Geometric Approach}
 \label{sec:fsint}

In this section, we first establish the structure of the fluid-solid configuration space as a principal fiber bundle and show that there exists a distinguished connection on this bundle.  We then recall
some general results from cotangent bundle reduction that will be useful to reduce the system and derive the equations of motion \eqref{EoM1}.


 \subsection{Geometric Fluid-Solid Interactions}

\paragraph{The Configuration Space.}  We describe the configurations of the body-fluid system by means of pairs $(g, \varphi)$, where $g$ is an element of $SE(2)$ describing the body motion and $\varphi: \mathcal{F}_0 \rightarrow \mathbb{R}^2$ is an embedding of the fluid reference space $\mathcal{F}_0$ into $\mathbb{R}^2$ describing the fluid.  The pairs $(g, \varphi)$  have to satisfy the impermeability boundary conditions dictated by the inviscid fluid model. The embedding $\varphi$ represents the configuration of an incompressible
  fluid and has therefore to be volume-preserving, i.e.,
$\varphi^\ast( dV ) = dV_0$, where $dV_0$ and $dV$ are volume forms on $\mathcal{F}_0$ and $\mathbb{R}^2$, respectively.   We denote the space of all such volume-preserving embeddings by $\mathrm{Emb}_{\mathrm{vol}}(\mathcal{F}_0, \mathbb{R}^2)$. We denote by $Q$ the space of all pairs $(g, \varphi)$, $g \in SE(2)$ and $\varphi \in  
\mathrm{Emb}_{\mathrm{vol}}(\mathcal{F}_0, \mathbb{R}^2)$ that satisfy the appropriate boundary 
conditions. That is, the configuration manifold $Q$ of the body-fluid system is a submanifold of the product space $SE(2) \times \mathrm{Emb}_{\mathrm{vol}}(\mathcal{F}_0, \mathbb{R}^2)$.

\paragraph{Tangent and Cotangent Spaces.} At each $(g, \varphi) \in Q$, the tangent space $T_{(g, \varphi)} Q$ 
is a subspace of  $T_g SE(2) \times T_{\varphi} \mathrm{Emb}_{\mathrm{vol}}(\mathcal{F}_0, \mathbb{R}^2)$ whose elements we denote by $(g, \varphi, \dot{g}, \dot{\varphi})$.  Here, $\dot{g}$ is an element of $T_g SE(2)$ and $\dot{\varphi}$ is a map from $\mathcal{F}_0$ to $T\mathbb{R}^2$ such that $\dot{\varphi}(x) \in T_{\varphi(x)} \mathbb{R}^2$ for all $x \in \mathcal{F}_0$.  Note that $\dot{g}$ represents the angular and linear velocity of the rigid body relative to the inertial frame while $\dot{\varphi}$ represents the \textbf{\emph{material}} or \textbf{\emph{Lagrangian velocity}} of the fluid. It is easier however to represent the elements of $TQ$ using 
the rigid body velocity expressed in body frame $\zeta$ and the fluid velocity $\mathbf{u}$.  
Note that, in the group theoretic  notation, the body velocity may be defined as $\zeta = g^{-1} \dot{g}$ and the \emph{\textbf{spatial}} or \emph{\textbf{Eulerian velocity field}} of the fluid may be defined as $\mathbf{u} := \dot{\varphi} \circ \varphi^{-1}$.  The vector field $\textbf{u}$ is a vector field on $\mathcal{F} := \varphi(\mathcal{F}_0)$, in contrast to $\dot{\varphi}$, which is merely a map from $\mathcal{F}_0$ to $T\mathbb{R}^2$.    We emphasize that $\zeta$ and $\mathbf{u}$ cannot be chosen arbitrarily, but have to satisfy the 
impermeability boundary conditions.  

The cotangent space $T^\ast_{(g, \varphi)} Q$ at a point $(g, \varphi) \in Q$ 
consists of elements $(g, \varphi, \pi, \alpha)$, where $\pi = \mathbb{M}_b \zeta \in \mathfrak{se}(2)^\ast$ 
is the momentum of the submerged body and $\alpha \in \Omega^1(\mathcal{F})$ is the one-form dual
of the velocity field $\mathbf{u}$, see~Appendices~\ref{appendix:rigidgroup} and~\ref{appendix:diffgroup}
for more details.

\paragraph{Kinetic Energy on $T^\ast Q$}. The kinetic energy in~\eqref{eq:T_solidfluid} can be used
to define a metric on the cotangent bundle $T^\ast Q$. To verify this, it is informative to recall here the 
\emph{\textbf{Hodge decomposition}} of differential forms (see \cite{AbMaRa1988}).  Any one-form $\alpha
\in \Omega^1(\mathcal{F})$ can be decomposed in a unique way as $\alpha = \mathbf{d}\Phi + \delta \Psi + \alpha_\Gamma$, where $\Phi \in \Omega^0(\mathcal{F})$, $\Psi \in \Omega^2(\mathcal{F})$, and $\alpha_\Gamma$ is a harmonic form: $\mathbf{d}\alpha_\Gamma = \delta \alpha_\Gamma = 0$.  
The  Hamiltonian of the solid-fluid system can now be defined as the function $H: T^\ast Q \rightarrow \mathbb{R}$ given by 
\begin{align*}
	H(g, \varphi, \pi_b, \alpha) & = 
		\frac{1}{2} \left\Vert \alpha\right\Vert^2 +
		\frac{1}{2} \pi^T \, \mathbb{M}_b^{-1} \, \pi 
		 = \frac{1}{2} \left\Vert \mathbf{d} \Phi_\zeta\right\Vert^2 + 
			\frac{1}{2} \left\Vert \delta \Psi\right\Vert^2 +
			 \frac{1}{2} \pi_b^T \, \mathbb{M}_b^{-1} \, \pi_b,
\end{align*}
where the norm in $\left\Vert \alpha\right\Vert^2$ is that in (\ref{formnorm}) on one-forms induced by the Euclidian metric and where we have again discarded the infinite constant $\left\Vert \alpha_\Gamma \right\Vert^2$. 
Here we have denoted the momentum by $\pi_b$ to emphasize the fact that this variable encodes the momentum of the rigid body only.  Later on, we will consider the momentum of the combined solid-fluid system, which will be denoted by $\pi = (\Pi, \mathbf{P})$. 

 The first term in 
the right-hand side of the Hamiltonian can be expressed in terms of the added mass matrix: 
\begin{align*}
	\left\Vert \mathbf{d} \Phi_\zeta\right\Vert^2 & = 
				\int_{\mathcal{F}} \mathbf{d} \Phi_\zeta \wedge \ast \mathbf{d} \Phi_\zeta 
		   = \oint_{\partial \mathcal{B}} \Phi_\zeta \ast \mathbf{d} \Phi_\zeta
		- \int_{\mathcal{F}} \Phi_\zeta \mathbf{d} \ast \mathbf{d} \Phi_\zeta \\
		& =  \oint_{\partial \mathcal{B}} \Phi_\zeta \frac{\partial \Phi_\zeta}{\partial n} \, dS 
			- \int_{\mathcal{F}} \Phi_\zeta \nabla^2 \Phi_\zeta \, dV
		 = \zeta^T \, \mathbb{M}_f \zeta,
\end{align*}
so that the Hamiltonian becomes 
\begin{equation} \label{almostredham}
H(g, \varphi, \pi_b, \alpha) = \frac{1}{2} \left\Vert \delta \Psi\right\Vert^2
	+ \frac{1}{2} \pi_b^T \, \mathbb{M}_b^{-1}(\mathbb{M}_b + \mathbb{M}_f)  \mathbb{M}_b^{-1} \, \pi_b
\end{equation}
using $\zeta = \mathbb{M}_b^{-1} \pi_b$.
The term involving $\delta \Psi$ yields the kinetic energy due to vortical structures present in the fluid
and is zero for the rigid body with circulation.

\paragraph{The Action of the Group of Volume-Preserving Diffeomorphisms.}  
 The group of all volume-preserving diffeomorphisms $\mathrm{Diff}_{\mathrm{vol}}(\mathcal{F})$
acts from the right on $Q$ by composition: 
for any $(g, \varphi) \in Q$ and $\phi \in \mathrm{Diff}_{\mathrm{vol}}(\mathcal{F})$ we define
\[
	(g, \varphi) \cdot \phi := (g, \varphi \circ \phi).
\]
This action leaves the kinetic energy \eqref{almostredham} on $Q$ invariant, since the Eulerian velocity field $\mathbf{u}$ is itself invariant.  The $\mathrm{Diff}_{\mathrm{vol}}(\mathcal{F})$-invariance
represents the \textbf{\emph{particle relabeling symmetry}}. 
The manifold $Q$ is hence the total space of a principal fiber bundle with structure group $\mathrm{Diff}_{\mathrm{vol}}(\mathcal{F})$ over $SE(2)$.  Here, the bundle projection projection $\mathrm{pr}: Q \rightarrow SE(2)$ is simply the projection onto the first factor: $\mathrm{pr}(g, \varphi) = g$.  One can readily show that the infinitesimal generator $X_Q$ corresponding to an element $X \in \mathfrak{X}_{\mathrm{vol}}(\mathcal{F}_0)$ is given by 
\begin{equation}
	X_Q(g, \varphi) := (0, \varphi_\ast X) \in T_{(g, \varphi)} Q.
\end{equation}

\paragraph{The Momentum Map of the Particle Relabeling Symmetry.}  
The group $\mathrm{Diff}_{\mathrm{vol}}(\mathcal{F}_0)$ acts on $Q$ and hence on $T^\ast Q$ by the cotangent lifted action.  We now compute the momentum map corresponding to this action, \emph{i.e.} the particle relabeling symmetry.  
This is a map $J$ from $T^\ast Q$ to $\mathfrak{X}^\ast_{\mathrm{vol}}(\mathcal{F}_0)$, and we recall from appendix~\ref{appendix:diffgroup} that $\mathfrak{X}^\ast_{\mathrm{vol}}(\mathcal{F}_0) = \mathbf{d}\Omega^1(\mathcal{F}_0) \times \mathbb{R}$.  Consequently, the momentum map has two components, corresponding with circulation and vorticity (pulled back to the reference configuration $\mathcal{F}_0$).  The statement that the momentum map is conserved then translates into \textbf{\emph{Kelvin's theorem}} that the vorticity is advected with the fluid, and that the circulation around each material loop is conserved.

\begin{proposition} \label{prop:vortmom}
The momentum map $J$ of the $\mathrm{Diff}_{\mathrm{vol}}(\mathcal{F}_0)$-action on $T^\ast Q$ is given by 
\begin{equation} \label{mommap}
	J(\pi_b, \alpha) = (\mathbf{d}\varphi^\ast\alpha, \Gamma), \quad \text{where} \quad \Gamma = \oint_{\partial \mathcal{B}} \varphi^\ast\alpha.
\end{equation}
\end{proposition}

\begin{proof}
We use the well-known formula for the momentum map of a cotangent lifted action 
(see \cite{AbMa78}).  For each $X \in \mathfrak{X}_{\mathrm{vol}}(\mathcal{F}_0)$, we have
\begin{align*}
	\left<J(\pi_b, \alpha), X\right> &= \left<( \pi_b, \alpha), X_Q(g, \varphi)\right>
		= \int_{\mathcal{F}} \left<\alpha, \varphi_\ast  X\right> dV 
		 = \int_{\mathcal{F}_0} \left<\varphi^\ast \alpha, X\right> dV_0,
\end{align*}
so that $J( \pi_b, \alpha) = [\varphi^\ast \alpha] \in \mathfrak{X}^\ast_{\mathrm{vol}}(\mathcal{F}_0)$.  After composition with the isomorphism (\ref{iso}) we obtain the desired form (\ref{mommap}).
\end{proof}

\paragraph{The One-Form $\alpha_\Gamma$ with Circulation $\Gamma$.}

Recall from section~\ref{sec:prel} that having circulation $\Gamma$ around the rigid body is equivalent to placing a point vortex of strength $\Gamma$ at the conformal center of the body, whose velocity field $\mathbf{u}_\Gamma$ was given in \eqref{circulation}.

In the remainder of this paper it will be easier to work with the one-form $\alpha_\Gamma$ on $\mathcal{F}$ given by $\alpha_\Gamma = \mathbf{u}_\Gamma^\flat$, or explicitly by 
\begin{equation} \label{vortalpha}
\alpha_\Gamma = \delta (\Psi_\Gamma dV),
\end{equation}
where $\Psi_\Gamma$ is the stream function \eqref{circulation}, and $dV$ is the volume form on $\mathbb{R}^2$.
\begin{proposition}
The one-form $\alpha_\Gamma$ is a harmonic one-form on $\mathcal{F}$ and satisfies 
\begin{equation} \label{contint}
	\oint_{\partial \mathcal{B}} \alpha_\Gamma = \Gamma.
\end{equation}
In particular, $\alpha_\Gamma$ is $L_2$-orthogonal to the space of exact one-forms.
\end{proposition}

\begin{proof}
The form $\alpha_\Gamma$ satisfies $\delta \alpha_\Gamma = 0$ by definition, and we have that $\mathbf{d} \alpha_\Gamma = \Gamma \delta(\mathbf{X}) dV$, so that $\mathbf{d} \alpha_\Gamma = 0$ in the fluid domain $\mathcal{F}$.  Hence, $\alpha_\Gamma$ is harmonic and (by means of the Hodge theorem) $L_2$-orthogonal to the space of exact one-forms.  The line integral \eqref{contint} follows from the expression \eqref{circulation} for the stream function $\Psi_\Gamma$.
\end{proof}

\paragraph{The Action of the Euclidian Symmetry Group.}
In addition to the right principal action of $\mathrm{Diff}_{\mathrm{vol}}(\mathcal{F}_0)$ on $Q$ described before, 
the special Euclidian group $SE(2)$ acts on $Q$ by bundle automorphisms from the \emph{left}.  In order words, there is an action $\psi : SE(2) \times Q \rightarrow Q$ given by 
\begin{equation} \label{seaction}
	\psi(h, (g, \varphi)) = h \cdot (g, \varphi) = (hg, h\varphi), 
\end{equation}
for all $h \in SE(2)$ and $(g, \varphi) \in Q$.  The embedding $h\varphi$ is defined by $(h\varphi)(x) = h \cdot \varphi(x)$, where the action on the right-hand side is just the standard action of $SE(2)$ on $\mathbb{R}^2$.  From a physical point of view, the $SE(2)$-action corresponds to the invariance of the combined solid-fluid system under arbitrary rotations and translations.

\paragraph{The Neumann Connection.}  

The bundle $\mathrm{pr} : Q \rightarrow SE(2)$ (defined by the $\mathrm{Diff}_{\mathrm{vol}}(\mathcal{F})$-invariance) is equipped with a principal fiber bundle connection, which was termed the \textbf{\emph{Neumann connection}} in \cite{VaKaMa2009}.  This connection seemed to have appeared first in  \cite{freeboundary} (see also \cite{Koiller1987}) and essentially encodes the effects of the rigid body on the ambient fluid.

The connection one-form $\mathcal{A}: TQ \rightarrow \mathfrak{X}_{\mathrm{vol}}(\mathcal{F}_0)$ of the Neumann connection is defined in terms of the  Helmholtz-Hodge decomposition \eqref{eq:u} of vector fields: if $(g, \varphi, \zeta, \mathbf{u})$ is an element of $T_{(g, \varphi)} Q$, then 
\begin{equation} \label{connform}
	\mathcal{A}(g, \varphi, \zeta, \mathbf{u}) = \varphi^\ast \mathbf{u}_{\mathrm{v}},
\end{equation}
where $\mathbf{u}_{\mathrm{v}}$ is the divergence-free part of the Eulerian velocity $\mathbf{u}$ in the Helmholtz-Hodge decomposition \eqref{eq:u}.  It can be shown (see \cite{VaKaMa2009}) that $\mathcal{A}$ satisfies the requirements of a connection one-form, and that $\mathcal{A}$ is invariant under the action of $SE(2)$ on $Q$:
\[
	\mathcal{A}_{(g, \varphi)}( T\psi_h(\zeta, \mathbf{u})) = \mathcal{A}_{(g, \varphi)}(\zeta, \mathbf{u}) 
\]
for all  $(\zeta, \mathbf{u}) \in T_{(g, \varphi)}Q$.  Here $\psi_h = \psi(h, \cdot)$, with $\psi$ the $SE(2)$-action \eqref{seaction}. Given the exact form of the Neumann connection, one can compute its curvature, which is a two-form $\mathcal{B}$ on the total space $Q$ with values in $\mathfrak{X}_{\mathrm{vol}}(\mathcal{F}_0)$. It turns out that there exists a closed-form formula for the curvature, which was first determined by \cite{MontgomeryThesis}
and further generalized by \cite{VaKaMa2009}.  More precisely, we compute an expression for the contraction $\left< \mu, \mathcal{B} \right>$, where $\mu$ is an arbitrary element of the dual space $\mathfrak{X}^\ast_{\mathrm{vol}}(\mathcal{F}_0)$.

\begin{proposition} \label{prop:curv} Let $(\zeta_1, \mathbf{u}_1)$ and $(\zeta_2, \mathbf{u}_2)$ be elements of
  $T_{(g, \varphi)} Q$ and denote the solutions of \eqref{Neumann} associated to $\zeta_1$ resp. $\zeta_2$ by $\Phi_1$ and $\Phi_2$.  Let $\mu$ be an arbitrary element of $\mathfrak{X}^\ast_{\mathrm{vol}}(\mathcal{F}_0)$. Then the $\mu$-component of
  the curvature $\mathcal{B}$ is given by
  \begin{equation} \label{curvature} \left<\mu, \mathcal{B}_{(g,
        \varphi)}((\zeta_1, \mathbf{u}_1), (\zeta_2, \mathbf{u}_2))\right> = \left\langle \! \left\langle\mu, \mathbf{d} \Phi_1 \wedge \mathbf{d}
      \Phi_2\right\rangle \!
  \right\rangle-\oint_{\partial \mathcal{B}} \alpha \wedge \ast(\mathbf{d}
    \Phi_1 \wedge \mathbf{d} \Phi_2),
  \end{equation}
  where $\left\langle \! \left\langle\cdot, \cdot\right\rangle \!
  \right\rangle$ is the metric on the space of
  forms on $\mathcal{F}$ defined in (\ref{formnorm}).
  \end{proposition}

In what follows, it will be necessary to have an expression for the curvature in terms of the elementary stream functions, rather than the elementary velocity potentials.  Such an expression can be easily obtained by noting that the stream function $\Psi$ is a harmonic conjugate to the velocity potential $\Phi$.  In particular, if $\Phi_1, \Phi_2$ are the velocity potentials introduced in the statement of proposition~\ref{prop:curv}, with their associated stream functions $\Psi_1, \Psi_2$, then 
$\mathbf{d}\Psi_1 \wedge \mathbf{d}\Psi_2 = \mathbf{d}\Phi_1 \wedge \mathbf{d}\Phi_2$.   The $\mu$-component of the curvature (\ref{curvature}) hence becomes 
\begin{align}
\left<\mu, \mathcal{B}\right> 
    & =  \left\langle \! \left\langle\mu, \mathbf{d} \Psi_1 \wedge \mathbf{d}
      \Psi_2\right\rangle \!
  \right\rangle  -\oint_{\partial \mathcal{B}} \alpha \wedge \ast(\mathbf{d}
    \Psi_1 \wedge \mathbf{d} \Psi_2). \label{curvaturestream}
\end{align}

We established the structure of the fluid-solid configuration space $Q$ as the total space of a principal fiber bundle with structure group $\mathrm{Diff}_{\mathrm{vol}}(\mathcal{F}_0)$.  We also  showed that there exists a distinguished connection on this bundle, which is invariant under the action of $SE(2)$ on $Q$.   In order to reduce the system and derive the equations of motion \eqref{EoM1}, we will need the framework of \emph{cotangent bundle reduction},  which we now describe.  More information and proofs of the results quoted below can be found in \cite{MaPe2000} and \cite{MarsdenHamRed}.


\subsection{Cotangent Bundle Reduction: Some General Results}

We describe cotangent bundle reduction in a general context.  Let the unreduced  configuration space be denoted by $Q$ and assume that a Lie group $G$ acts freely and properly on $Q$ from the right, so that the quotient space $Q/G$ is a manifold.   Furthermore, we assume that the quotient space $Q/G$ is equal to a second Lie group $H$, which acts on $Q$ from the left.  In Section~\ref{sec:red}, $Q$ will be the fluid-solid configuration space, $G$ will be the group 
$\mathrm{Diff}_{\mathrm{vol}}(\mathcal{F}_0)$ and $H$ will be the special Euclidian group $SE(2)$.  Our first goal is to reduce by the $G$-action and describe the reduced phase space.   This is described in theorem~\ref{thm:cotbundle}.  In the second stage of the reduction, we will then do Poisson reduction with respect to the residual group $H$ in order to obtain a Poisson structure on the twice reduced space.  This is described below in theorem~\ref{thm:poisson}.

Recall that the quotient projection $\mathrm{pr}_{Q, G}: Q \rightarrow Q/G$ defines a right principal fiber bundle, and assume that a connection on this fiber bundle is given, with connection one-form $\mathcal{A}: TQ \rightarrow \mathfrak{g}$.  For more information about principal fiber bundles and connections, see \cite{KN1}.
The group $G$ acts on $T^\ast Q$ by cotangent lifts, and we denote the momentum map of this action by $J: T^\ast Q\rightarrow \mathfrak{g}^\ast$. 
%
The next theorem characterizes the reduced phase space $(T^\ast Q)_\mu := J^{-1}(\mu)/G_\mu$, where $\mu$ is an arbitrary element of $\mathfrak{g}^\ast$.  Here, $G_\mu$ is the isotropy subgroup of $\mu$, defined as follows: $g \in G_\mu$ if $\mathrm{Ad}_g^\ast \mu = \mu$, where $\mathrm{Ad}_g^\ast$ denotes the coadjoint action of $G$ on $\mathfrak{g}^\ast$.  The reduced phase space $(T^\ast Q)_\mu$ can be described in full generality, but we focus here on the special case where the isotropy group $G_\mu$ is the full symmetry group: $G_\mu = G$.  We will see in section~\ref{sec:diffred} that  this is the relevant case to consider for the rigid body with circulation.

\begin{theorem} \label{thm:cotbundle}
	Let $G$ be a group acting freely and properly from the right on a manifold $Q$ so that $\mathrm{pr}_{Q, G}: Q \rightarrow Q/G$ is a principal fiber bundle.  Let $\mathcal{A}: TQ \rightarrow \mathfrak{g}$ be a connection one-form on this bundle.   Let $\mu \in \mathfrak{g}^\ast$ and assume that $G_\mu = G$.  
	
	Then there is a symplectic diffeomorphism between $(T^\ast Q)_\mu$ and $T^\ast (Q/G)$, the latter with symplectic form $\omega_{\mathrm{can}} - B_\mu$; here $\omega_{\mathrm{can}}$ is the canonical symplectic form on $T^\ast (Q/G)$ and $B_\mu = \mathrm{pr}^\ast_{Q/G} \beta_\mu$, where 
$\mathrm{pr}_{Q/G}: T^\ast (Q/G) \rightarrow Q/G$ is the cotangent bundle projection, and $\beta_\mu$ 
is determined through	
	\begin{equation} \label{magform}
		\mathrm{pr}^\ast_{Q, G} \beta_\mu = \mathbf{d}\left\langle \mu, \mathcal{A}\right\rangle.
	\end{equation}
\end{theorem}
\begin{proof}[\textrm{\textbf{Outline of the Proof}}]
This is a special case of theorem~2.3.3 in \cite{MarsdenHamRed}.  We just recall the explicit form of the isomorphism between $(T^\ast Q)_\mu$ and $T^\ast (Q/G)$; the proof that this map also preserves the relevant symplectic structures can be found in \cite{MarsdenHamRed}.

The isomorphism $\varphi_\mu : (T^\ast Q)_\mu \rightarrow T^\ast (Q/G)$ is the composition of the map ${\mathrm{shift}}_\mu : (T^\ast Q)_\mu  \rightarrow (T^\ast Q)_0$ and the map $\varphi_0:  (T^\ast Q)_0 \rightarrow T^\ast (Q/G)$:
\begin{equation} \label{rediso}
	\varphi_\mu = \varphi_0 \circ {\mathrm{shift}}_\mu.
\end{equation}
Both of these constitutive maps are isomorphisms.   The map ${\mathrm{shift}}_\mu$ is defined as follows: we first introduce a map ${\mathrm{Shift}}_\mu : J^{-1}(\mu) \rightarrow J^{-1}(0)$ by 
\[
	{\mathrm{Shift}}_\mu(\alpha_q) = \alpha_q - \left< \mu, \mathcal{A}(q) \right>.
\]
It can easily be verified that ${\mathrm{Shift}}$ is $G$-invariant so that it drops to a quotient map ${\mathrm{shift}}_\mu: (T^\ast Q)_\mu \rightarrow (T^\ast Q)_0$.

Secondly, the map $\varphi_0: (T^\ast Q)_0 \rightarrow T^\ast (Q/G)$ is defined by noting that 
\[
	J^{-1}(0) = \{  \alpha_q \in T^\ast Q: \left< \alpha_q , \xi_Q(q) \right> = 0 \quad \text{for all $\xi \in \mathfrak{g}$} \}
\]
so that the map $\bar{\varphi}_0 : J^{-1}(0) \rightarrow T^\ast(Q/G)$ given by 
\begin{equation} \label{barphi}
	\left< \bar{\varphi}_0(\alpha_q) , T \pi_{Q, G}(v_q) \right> = \left<\alpha_q, v_q \right> 
\end{equation}
is well defined.  The map $\bar{\varphi}_0$ is easily seen to be $G$-invariant and surjective, and hence induces a quotient map $\varphi_0: (T^\ast Q)_0 \rightarrow T^\ast (Q/G)$.
\end{proof}

Note that the isomorphism between $(T^\ast Q)_\mu$ and $T^\ast Q/G$ is connection-dependent.  As a result, the reduced symplectic form on $T^\ast Q/G$ is modified by the two-form $\beta_\mu$, which is traditionally referred to as a \emph{\textbf{magnetic term}} since it also appears in the description of a charged particle in a magnetic field (see \cite{GuSt1984}).

Having described the reduced phase space $(T^\ast Q)_\mu$, we now take into account the assumption made earlier that the base space $Q/G$ has the structure of a second Lie group $H$, which acts on $Q$ from the \emph{left} and leaves the connection one-form $\mathcal{A}$ invariant.  In this case, the reduced phase space $(T^\ast Q)_\mu$ is equal to $T^\ast H$, equipped with the magnetic symplectic structure described before.  As a result, $H$ acts on $T^\ast H$, and it can be checked that this action leaves the symplectic structure invariant.  It would now be possible to do symplectic reduction as above for the $H$-action as well to obtain a fully reduced symplectic structure.  However, all that is needed in section~\ref{sec:eucsymm} is an expression for the reduced Poisson structure, described in the following theorem.

\begin{theorem}[Theorem~7.2.1 in \cite{MarsdenHamRed}] \label{thm:poisson}
  The Poisson reduced space for the left cotangent lifted action of
  $H$ on $(T^\ast H, \omega_{\mathrm{can}} - B_\mu)$ is $\mathfrak{h}^\ast$ with
  Poisson bracket given by
\begin{equation} \label{bracket}
\{f, g\}_\mathcal{B}(\mu) = -\left< \mu, \left[ \frac{\delta f}{\delta \mu},
    \frac{\delta g}{\delta \mu} \right] \right> - {B}_\mu(e)\left(
  \frac{\delta f}{\delta \mu}, \frac{\delta g}{\delta \mu} \right)
\end{equation}
for $f, g \in C^\infty(\mathfrak{h}^\ast)$.
\end{theorem}

The theorem in \cite{MarsdenHamRed} is proved for right actions,
whereas the action of $H$ here is assumed to be from the left.   However, the
same proof continues to hold, \emph{mutatis mutandis}.  

Lastly, we recall the definition of the $\mathcal{B}\mathfrak{h}$-potential (see \cite{MarsdenHamRed}), which plays the role of momentum map relative to the magnetic term.  While this potential can be defined on an arbitrary manifold with a magnetic symplectic form, we treat just the case of a Lie group $H$ with left-invariant symplectic form $\omega_{\mathrm{can}} - B_\mu$ as before.

\begin{definition}  \label{def:bgpot} Let $H$ be a Lie group with left-invariant symplectic form $\omega_{\mathrm{can}} - B_\mu$.
	Suppose there exists a smooth map $\psi: H \rightarrow \mathfrak{h}^\ast$ such that 
	\begin{equation} \label{defbgpot}
		\mathbf{i}_{\xi_H} B_\mu = \mathbf{d} \left<\psi, \xi \right>,
	\end{equation}
	for all $\xi \in \mathfrak{h}$.  Then the map $\psi$ is called the $\mathcal{B}\mathfrak{h}$-potential of the $H$-action, relative to the magnetic term $B_\mu$.
\end{definition}


In what follows, we will always assume that $\psi$ exists.  In this case, $\psi$ is defined up to an arbitrary constant, which we normalize by assuming that $\psi(e) = 0$.  Under this assumption,  the \emph{non-equivariance one-cocycle} $\sigma: H \rightarrow \mathfrak{h}$ associated to $\psi$,
\begin{equation} \label{onecocycle}
	\sigma(g) = \psi(g) - \mathrm{Ad}^\ast_{g^{-1}} \psi(e), 
\end{equation}
coincides with $\psi$.  We will make no further distinction between $\psi$ and $\sigma$.  The importance of $\psi$ lies in the fact that we may characterize the symplectic leaves of the magnetic bracket \eqref{bracket} in $\mathfrak{h}^\ast$ as orbits of a certain affine action of $H$ on $\mathfrak{h}^\ast$.
\begin{proposition} \label{prop:afforbit}
	Let  $(T^\ast H, \omega_{\mathrm{can}} - B_\mu)$ be a Lie group with a magnetic symplectic form.
	The symplectic leaves in $\mathfrak{h}^\ast$ of the magnetic Poisson structure  \eqref{bracket} are the orbits of the  following affine action:
	\begin{equation} \label{affaction}
		g \cdot \mu = \mathrm{Ad}^\ast_{g^{-1}} \mu + \psi(g),
	\end{equation}
	where $g \in H, \mu \in \mathfrak{h}^\ast$.
\end{proposition}
\begin{proof}
	See theorem~7.2.2 in \cite{MarsdenHamRed}.
\end{proof}
%

\section{Derivation of the Chaplygin-Lamb Equations via Reduction by Stages}
\label{sec:red}

\subsection{Reduction by the Diffeomorphism Group} 
\label{sec:diffred}

We impose the condition that the fluid has constant circulation $\Gamma$ by considering the symplectic reduced space $J^{-1}(0, \Gamma)/\mathrm{Diff}_{\mathrm{vol}}(\mathcal{F}_0)$.  Our goal is to use cotangent bundle
reduction to express this space in a more manageable form  (via establishing an isomorphism that is connection-dependent) and to obtain an explicit expression for the reduced symplectic form on this space.

The reduction theorem~\ref{thm:cotbundle} deals with the case where the isotropy group coincides with the full group.  It can  be easily verified that this condition is satisfied for the rigid body with circulation: here, the relevant momentum value is $(0, \Gamma)$, and by \eqref{coad} we have that for all $\phi \in \mathrm{Diff}_{\mathrm{vol}}(\mathcal{F}_0)$, 
\[
	\mathrm{CoAd}_\phi(0, \Gamma) = (0, \Gamma), 
\]
so that $(\mathrm{Diff}_{\mathrm{vol}}(\mathcal{F}_0))_{(0, \Gamma)} = \mathrm{Diff}_{\mathrm{vol}}(\mathcal{F}_0)$.

\paragraph{The Reduced Phase Space.}
As an introduction to the methods of this section, we establish an isomorphism $\varphi_\Gamma$ between the reduced phase space $J^{-1}(0, \Gamma)/\mathrm{Diff}_{\mathrm{vol}}(\mathcal{F}_0)$ and $T^\ast SE(2)$ in a relatively ad-hoc manner.   We then show that cotangent bundle reduction yields precisely this isomorphism.

Let us first introduce the linear isomorphism  
$\mathfrak{m} : \mathfrak{se}(2)^\ast \rightarrow \mathfrak{se}(2)^\ast$ given by 
\begin{align} \label{mapm}
	\mathfrak{m}(\pi_b) & = \pi_b + \mathbb{M}_f \mathbb{M}_b^{-1} \pi_b 
		 =  ( \mathbb{M}_b + \mathbb{M}_f ) \mathbb{M}_b^{-1} \pi_b
\end{align}	
where $\mathbb{M}_b$ and $\mathbb{M}_f$ are the body mass matrix and the added mass matrix, respectively.   The map $\mathfrak{m}$ will prove to be crucial later on; its effect is to redefine the momentum of the rigid body in order to take into account the added mass effects.

We now derive an explicit expression for the level sets $J^{-1}(0)$ and $J^{-1}(0, \Gamma)$ of the vorticity momentum map.  Let $(g, \varphi, \pi_b, \alpha)$ be an element of $T^\ast Q$: the requirement that $J(\pi_b, \alpha) = (0, \Gamma)$ is equivalent to
\[
	\mathbf{d} \alpha = 0 \quad \text{and} \quad \oint_{\partial \mathcal{B}} \varphi^\ast \alpha = \Gamma.
\]
By means of the Hodge decomposition, this implies that 
\[
	\alpha = \alpha_\Gamma + \mathbf{d}\Phi_\zeta, 
\]
where $\alpha_\Gamma$ is given by \eqref{vortalpha} and $\Phi_\zeta$ is the solution of the Neumann problem \eqref{Neumann} with boundary data $\zeta = \mathbb{M}_b^{-1} \pi_b$. Hence, we may define an isomorphism $\psi_\Gamma$ from $J^{-1}(0, \Gamma)$ to $Q \times \mathfrak{se}(2)^\ast$, given by 
\[
	\psi_\Gamma: (g, \varphi, \pi_b, \alpha) \mapsto (g, \varphi, \mathfrak{m}(\pi_b)) \in Q \times \mathfrak{se}(2)^\ast,
\]
where $\mathfrak{m}$ is the map \eqref{mapm}. Likewise, there is an isomorphism $\psi_0 : J^{-1}(0) \rightarrow Q \times \mathfrak{se}(2)^\ast$, obtained by realizing that $ J^{-1}(0)$ consists of elements of the form $(g, \varphi, \pi_b, \mathbf{d} \Phi_\zeta)$, where $\Phi_\zeta$ has the same interpretation as above.

The map $\psi_\Gamma$ is easily seen to be $\mathrm{Diff}_{\mathrm{vol}}(\mathcal{F}_0)$-equivariant and hence drops to a quotient isomorphism $\varphi_\Gamma$ between $ J^{-1}(0, \Gamma)/\mathrm{Diff}_{\mathrm{vol}}(\mathcal{F}_0) $ and $T^\ast SE(2)$, given by 
\[
		\varphi_\Gamma: [(g, \varphi, \pi_b, \alpha)] \mapsto (g, \mathfrak{m}(\pi_b)).
\]
We see that the effect of the map $\varphi_\Gamma$ is to eliminate the influence of the fluid entirely, except for the added mass effects, which are encoded by $\mathfrak{m}$.

\paragraph{The $\Gamma$-Component of the Connection One-Form.}
In order to compute the magnetic term \eqref{magform}, we need an expression for the connection one-form, contracted with the momentum value at which we do reduction.  For the rigid body with circulation, the connection is the Neumann connection, and the momentum value is $(0, \Gamma)$.  Therefore, let $\mathcal{A}$ be the connection one-form of the Neumann connection as in \eqref{connform} and define the $\Gamma$-component of $\mathcal{A}$ to be the one-form $\mathcal{A}_\Gamma: TQ \rightarrow \mathbb{R}$ given by 
\[
	\mathcal{A}_\Gamma(g, \varphi, \zeta, \mathbf{u}) := 
	\left< (0, \Gamma), \mathcal{A}(g, \varphi, \zeta, \mathbf{u})\right>, 
\]
where on the right-hand side we interpret $(0, \Gamma)$ as an element of $\mathfrak{X}_{\mathrm{vol}}(\mathcal{F}_0)^\ast$.   In other words, 
$\mathcal{A}_\Gamma(g, \varphi, \zeta, \mathbf{u})$ computes the divergence-free part of $\mathbf{u}$ and contracts it with the element $(0, \Gamma)$ of  $\mathfrak{X}_{\mathrm{vol}}(\mathcal{F}_0)^\ast$.  This hints at the fact that $\mathcal{A}_\Gamma$ is nothing but the one-form $\alpha_\Gamma$ of \eqref{vortalpha}, which we prove in the next proposition.

\begin{proposition} \label{prop:compalpha}
	The $\Gamma$-component of the connection one-form $\mathcal{A}$ is related to the form $\alpha_\Gamma$ by the following relation: for all $(g, \varphi) \in Q$ and $(\zeta, \mathbf{u}) \in T_{(g, \varphi)} Q$, we have that
	\begin{equation} \label{Aalpha}
		\mathcal{A}_\Gamma(\zeta, \mathbf{u}) = \int_{\mathcal{F}} \alpha_\Gamma(\mathbf{u}) \, dV.
	\end{equation}
\end{proposition}
\begin{proof}
By definition, we then have that $\mathcal{A}_\Gamma$ is given by 
\[
	\mathcal{A}_\Gamma(\zeta, \mathbf{u})  =  \left< (0, \Gamma), \mathcal{A}(\zeta, \mathbf{u}) \right>	
		 = \left<\alpha_\Gamma, \mathbf{u}_{\mathrm{v}} \right>  = \int_{\mathcal{F}} \alpha_\Gamma(\mathbf{u}) \, dV, 
\]
where we have used the 
Helmholtz-Hodge decomposition $\mathbf{u} = \mathbf{u}_{\mathrm{v}} + \nabla \Phi_\zeta$, 
together with the
fact that $\alpha_\Gamma$ satisfies \eqref{contint} and is $L_2$-orthogonal to gradient vector fields.  
\end{proof}

\paragraph{The $\Gamma$-Component of the Curvature.}  As a second step towards the computation of the magnetic symplectic form \eqref{magform}, we need a convenient expression for the exterior derivative $\mathbf{d} \mathcal{A}_\Gamma$.  While it is in theory possible to compute the derivative directly from the expression \eqref{Aalpha} for $\mathcal{A}_\Gamma$, we can avoid much of the technicalities by means of the following insight: for the rigid body with circulation, $\mathbf{d} \mathcal{A}_\Gamma$ is nothing but the $\Gamma$-component of the curvature of the Neumann connection, as defined below.  Recall that an explicit expression for the latter was established in proposition~\ref{prop:curv}. We define the $\Gamma$-component of the curvature $\mathcal{B}$ as before by the following prescription:
\[
	\mathcal{B}_\Gamma((\zeta, \mathbf{u}),  (\xi, \mathbf{v})) :=
		\left< (0, \Gamma), \mathcal{B}((\zeta, \mathbf{u}),  (\xi, \mathbf{v})) \right>,
\]
again relying on the fact that $(0, \Gamma) \in \mathfrak{X}_{\mathrm{vol}}(\mathcal{F}_0)^\ast$, while $\mathcal{B}((\zeta, \mathbf{u}),  (\xi, \mathbf{v}))$ is an element of $\mathfrak{X}_{\mathrm{vol}}(\mathcal{F}_0)$.

\begin{proposition} \label{prop:gammacurv}
	The $\Gamma$-component of the curvature is 
	a left $SE(2)$-invariant two-form on $SE(2)$ given at the identity by
\begin{equation}\label{curvB}
	\mathcal{B}_{\Gamma}(e) = \Gamma \mathbf{e}_x^\ast \wedge \mathbf{e}_y^\ast.
\end{equation}
\end{proposition}
\begin{proof}
The left-invariance of $\mathcal{B}_{\Gamma}$ follows from the fact that the Neumann connection itself is left invariant under the action of $SE(2)$ upon itself.  We may hence restrict our attention to the value of the curvature at a point $(e, \varphi) \in Q$, where $e$ is the identity in $SE(2)$.   The fact that $\mathcal{B}_{\Gamma}$ drops to $SE(2)$ will be obvious once we determine its exact form, but can also be proved directly; see for instance \cite{VaKaMa2009}.

The first integral in the expression (\ref{curvaturestream}) for the curvature vanishes since the integration is over the fluid domain, where $\mathbf{d}\alpha_\Gamma = 0$.  The second integral can be computed explicitly, since we only need the expression (\ref{boundary}) for the stream function on the boundary.

Let $(g, \varphi, \zeta_i, \mathbf{u}_i)$, $i = 1, 2$, be elements of $T_{(g, \varphi)} Q$ and write the velocity as $\zeta_i = (\Omega_i, \mathbf{V}_i)$, where $\mathbf{V}_i = U_i(\cos \alpha_i, \sin \alpha_i)$.  Consider the stream functions $\Psi_i$ given in \eqref{boundary} corresponding to the rigid body motions $(\Omega_i, \mathbf{V}_i)$.  On the boundary $\partial \mathcal{B}$, we then have that
\begin{align*}
	\ast \left( \mathbf{d}\Psi_1 \wedge \mathbf{d}\Psi_2 \right) =  
		U_1 U_2 \sin(\alpha_1 - \alpha_2) & +  \Omega_1 U_2 (x \cos\alpha_2 + y \sin\alpha_2) \\
			& - \Omega_2 U_1 (x\cos\alpha_1 + y \sin\alpha_1),
\end{align*}
so that the curvature is given by
\[
	\mathcal{B}_{\Gamma}((\zeta_1, \mathbf{u}_1), (\zeta_2, \mathbf{u}_2)) = 
	- \int _{\partial \mathcal{B}} \alpha_\Gamma \wedge \ast(\mathbf{d}\Psi_1 \wedge \mathbf{d}\Psi_2) = -\Gamma U_1 U_2 \sin(\alpha_1 - \alpha_2).
\]
Here, we have used the fact that 
\[
	\oint_{\partial\mathcal{B}} \alpha_\Gamma = \Gamma \quad \text{and} \quad \oint_{\partial\mathcal{B}} x \alpha_\Gamma = \oint_{\partial\mathcal{B}} y \alpha_\Gamma = 0.
\]
Note that the value of $\mathcal{B}_{\Gamma}$ does not depend on $\mathbf{u}_1, \mathbf{u}_2$ so that $\mathcal{B}_{\Gamma}$ drops to $SE(2)$ as claimed.  Finally, using the basis (\ref{dualbasis}) of $\mathfrak{se}(2)^\ast$, we may write the curvature at the identity as 
\[
	\mathcal{B}_{\Gamma}(e) = \Gamma \mathbf{e}_x^\ast \wedge \mathbf{e}_y^\ast.
\]
This concludes the proof.
\end{proof}

We now prove the main result of this section, that the exterior derivative $\mathbf{d}\mathcal{A}_\Gamma$ is equal to the $\Gamma$-component $\mathcal{B}_\Gamma$ of the curvature.  More precisely, recalling the projection $\mathrm{pr} : Q \rightarrow SE(2)$, we have 
\begin{equation} \label{curveq}
	\mathrm{pr}^\ast \mathcal{B}_\Gamma = \mathbf{d} \mathcal{A}_\Gamma.
\end{equation}

Note that this is not true for arbitrary reduced cotangent bundles, and that this 
is highly specific to the case of rigid bodies with circulation.  The expression \eqref{curveq} can be proved as follows:  because of the Cartan structure formula $\mathcal{B} = \mathbf{d}\mathcal{A} + [\mathcal{A}, \mathcal{A}]$ for right actions, we have that 
\begin{equation} \label{csf}
\mathbf{d} \mathcal{A}_\Gamma =  \mathrm{pr}^\ast \mathcal{B}_\Gamma -\left\langle (0, \Gamma),  [\mathcal{A}, \mathcal{A}]\right\rangle.
\end{equation}
It remains to show that the second term on the right-hand side vanishes.   It can be shown (see the proof of 
theorem~4.2 in \cite{mw1983}) that, for divergence-free
vector fields $\mathbf{u}_{\mathrm{v}, 1}$ and $\mathbf{u}_{\mathrm{v}, 2}$ tangent to $\partial \mathcal{F}$
and arbitrary one-forms $\alpha$, the following holds: 
\begin{equation} \label{magic}
  \int_\mathcal{F} \alpha([\mathbf{u}_{\mathrm{v}, 1}, \mathbf{u}_{\mathrm{v}, 2}]) \,d V = 
  \int_\mathcal{F} \mathbf{d} \alpha(\mathbf{u}_{\mathrm{v}, 1}, \mathbf{u}_{\mathrm{v}, 2}) \,d V.
\end{equation}
If we now put $\mathbf{u}_{\mathrm{v}, i} := \mathcal{A}(\zeta_i, \mathbf{u}_i)$, $i = 1,2$, we may rewrite the second term in (\ref{csf}) as follows:
%
%
\begin{align*}
\left\langle (0, \Gamma), [\mathcal{A}(\zeta_1, \mathbf{u}_1), \mathcal{A}(\zeta_2, \mathbf{u}_2)] \right\rangle  & = \left\langle (0, \Gamma), [\mathbf{u}_{\mathrm{v},1}, \mathbf{u}_{\mathrm{v},2}] \right\rangle \\
	&  = \int_{\mathcal{F}} \alpha_\Gamma([\mathbf{u}_{\mathrm{v},1}, \mathbf{u}_{\mathrm{v},2}]) \,dV 
	  = \int_{\mathcal{F}} \mathbf{d}\alpha_\Gamma(\mathbf{u}_{\mathrm{v},1}, \mathbf{u}_{\mathrm{v},2}) \,dV,\end{align*}

where $\alpha_\Gamma$ is the one-form given by (\ref{vortalpha}).  However, the integral on the right-hand side is again zero because the integration is over $\mathcal{F}$, while the support of $\mathbf{d}\alpha_\Gamma$ is concentrated inside the rigid body.
This shows that $\mathbf{d}\mathcal{A}_\Gamma$ in (\ref{csf}) is nothing but the $\Gamma$-component of the curvature of the Neumann connection, as calculated in proposition~\ref{prop:gammacurv}.  


\paragraph{The Reduced Phase Space.}  
We formally establish the isomorphism $\varphi_\Gamma$ of theorem~\ref{thm:cotbundle} between the reduced configuration space $J^{-1}(0, \Gamma)/\mathrm{Diff}_{\mathrm{vol}}(\mathcal{F}_0)$ and $T^\ast SE(2)$.   Recall that $\varphi_\Gamma$ is defined as the composition $\varphi_0 \circ {\mathrm{shift}}_\Gamma$, where $\varphi_0$ and ${\mathrm{shift}}_\Gamma$ were defined in the proof of theorem~\ref{thm:cotbundle}.  The following two lemmas are devoted to the explicit form of these two constitutive maps for the rigid body with circulation.

\begin{lemma}
	Consider the map $\mathfrak{m} : \mathfrak{se}(2)^\ast \rightarrow \mathfrak{se}(2)^\ast$ 
	defined in \eqref{mapm}. The map $\varphi_0: (T^\ast Q)_0 \rightarrow T^\ast SE(2)$ defined 
	in the proof of theorem~\ref{thm:cotbundle} is given in terms of $\mathfrak{m}$ by the following expression:
	\begin{equation} \label{multmap}
		\varphi_0( [g, \varphi, \pi_b, \mathbf{d}\Phi_\zeta]) = (g, \mathfrak{m}(\pi_b)). 
	\end{equation}
\end{lemma}
\begin{proof}
Consider an element $(g, \varphi, \pi_b, \mathbf{d}\Phi_\zeta)$ of $J^{-1}(0)$.  
In a left trivialization, 
the map $\bar{\varphi}_0: J^{-1}(0) \rightarrow T^\ast SE(2)$ defined in \eqref{barphi} is given  by 
\begin{align*}
	\left< \bar{\varphi}(\pi_b, \mathbf{d}\Phi_\zeta), \xi \right> 
	 = \left< (\pi_b, \mathbf{d}\Phi_\zeta), (\xi, \nabla \Phi_\xi) \right> 
	= \left< \pi_b,\xi \right> + \left<\mathbf{d}\Phi_\zeta, \nabla \Phi_\xi \right>,
\end{align*}
where now $\Phi_\xi$ is the solution of \eqref{Neumann} with boundary data $\xi$.  
The right-hand side can then be written as follows, using the definition of the added-mass matrices:
\begin{align*}
\left< (\pi_b, \mathbf{d}\Phi_\zeta), (\xi, \nabla \Phi_\xi) \right>
& = \left<\pi_b, \xi \right> + \int_{\mathcal{F}} \mathbf{d}\Phi_\zeta \cdot \nabla \Phi_\xi \, dV 
 = \zeta^T (\mathbb{M}_b +  \mathbb{M}_f )\xi \\
&  = \left< (\mathbb{I} + \mathbb{M}_f \mathbb{M}_b^{-1}) \pi_b, \xi \right>
\end{align*}
since $\pi_b = \mathbb{M}_b \zeta$.  In other words, we have that 
\[
	\left< \bar{\varphi}(g, \varphi, \pi_b, \mathbf{d}\Phi_\zeta), \xi\right>
	= \left< \mathfrak{m}(\pi_b), \xi \right>, 
\]
so that $\bar{\varphi}(g, \varphi, \pi_b, \mathbf{d}\Phi_\zeta) = (g,\mathfrak{m}(\pi_b))$.
\end{proof}


We now determine the map ${\mathrm{shift}}_\Gamma : (T^\ast Q)_{(0, \Gamma)} \rightarrow (T^\ast Q)_0$.  Because of proposition~\ref{prop:compalpha}, the map ${\mathrm{Shift}}_\Gamma : J^{-1}(0, \Gamma) \rightarrow J^{-1}(0)$ is in our case a simple shift by $\alpha_\Gamma$: 
\[
	{\mathrm{Shift}}_\Gamma(g, \varphi, \pi_b, \alpha) = (g, \varphi, \pi_b, \alpha-\alpha_\Gamma).
\]
Physically speaking, we may think of $\alpha \in J^{-1}(0, \Gamma)$ as a one-form with circulation $\Gamma$; by subtracting $\alpha_\Gamma$ from $\alpha$, we obtain a one-form with zero circulation.  For the quotient map 
${\mathrm{shift}}_\Gamma: (T^\ast Q)_{(0, \Gamma)} \rightarrow (T^\ast Q)_0$ we then have the following result.

\begin{lemma}
The quotient map 
${\mathrm{shift}}_\Gamma: (T^\ast Q)_{(0, \Gamma)} \rightarrow (T^\ast Q)_0$ is given by 
\[
	{\mathrm{shift}}_\Gamma([(g, \varphi, \pi_b, \alpha)]) = [(g, \varphi, \pi_b, \alpha-\alpha_\Gamma)].
\]
\end{lemma}

By concatenating these two maps, it is now straightforward to find the explicit form of the isomorphism $\varphi_\Gamma$  between $(T^\ast Q)_{(0, \Gamma)}$ and $T^\ast SE(2)$:  from \eqref{rediso} we have that the isomorphism is given by 
\begin{equation} \label{dirmap}
	\varphi_\Gamma: [(g, \varphi, \pi_b, \alpha)] \mapsto (g, \mathfrak{m}(\pi_b)).
\end{equation}
The general theory of cotangent bundle reduction guarantees that this map is an isomorphism, but 
it is instructive to construct the inverse mapping explicitly.  For every $(g, \pi) \in T^\ast SE(2)$, put 
$\pi_b := \mathfrak{m}^{-1}(\pi)$ and set $\alpha := \alpha_\Gamma + \mathbf{d}\Phi_\zeta$, where $\Phi_\zeta$ is the solution of the Neumann problem \eqref{Neumann} with boundary data $\zeta = \mathbb{M}_b^{-1} \pi_b$.  Furthermore, choose an arbitary fluid embedding $\varphi$ such that $(g, \varphi) \in Q$.  The inverse mapping $\varphi_\Gamma^{-1}$ is then given by 
\begin{equation} \label{invmap}
	\varphi_\Gamma^{-1}(g, \pi) = [(g, \varphi, \pi_b, \alpha)], \quad
	\text{where $\pi = \mathfrak{m}^{-1}(\pi_b)$ and $\alpha = \alpha_\Gamma + \mathbf{d}\Phi_\zeta$.}
\end{equation}
It is straightforward to check that $\varphi_\Gamma^{-1} \circ \varphi_\Gamma = \varphi_\Gamma \circ \varphi_\Gamma^{-1} = \mathrm{id}$.

\paragraph{The Reduced Hamiltonian.}

As a last step towards establishing a reduced Hamiltonian formulation for the rigid body with circulation, we need to find an appropriate expression for the Hamiltonian on this space.  This can be done by computing first the Hamiltonian on the unreduced space from the kinetic energy (\ref{Tkin}) and then noting that the result induces a well-defined function on the reduced phase space.  

\begin{proposition}
	The Hamiltonian on the reduced phase space $T^\ast SE(2)$ is given by 
	\begin{equation} \label{redhamil}
		H_{\mathrm{red}}(g, \pi) =  \frac{1}{2} \pi^T\, \mathbb{M}^{-1} \,\pi,
	\end{equation}
	where $\mathbb{M} = \mathbb{M}_b + \mathbb{M}_f$.
\end{proposition}
\begin{proof}
The proof relies on the explicit form of the isomorphism $\varphi_\Gamma$ in \eqref{dirmap}.  Let $(g, \pi)$ be an element of $T^\ast SE(2)$ and consider an arbitrary element $(g, \varphi, \pi_b, \alpha) \in J^{-1}(0, \Gamma)$ such that 
$\varphi_\Gamma^{-1}(g, \pi) = [(g, \varphi, \pi_b, \alpha)]$.  Recall from the discussion following \eqref{invmap} that $\pi_b$ and $\alpha$ are given by 
\begin{equation} \label{momrel}
	\pi_b = \mathfrak{m}^{-1}(\pi)
	\quad \text{and} \quad 
	\alpha = \alpha_\Gamma + \mathbf{d}\Phi_\zeta,
\end{equation}
where $\Phi_\zeta$ is the solution of the Neumann problem \eqref{Neumann} with boundary data $\zeta = \mathbb{M}_b^{-1} \pi_b$.  The relation between the reduced Hamiltonian $H_{\mathrm{red}}$ on $T^\ast SE(2)$ and the Hamiltonian \eqref{almostredham} is then written as
\begin{align*}
	H_{\mathrm{red}}(g, \pi) & = H(g, \varphi, \pi_b, \alpha) 
		 =  \frac{1}{2} \pi_b^T \, \mathbb{M}_b^{-1}(\mathbb{M}_b + \mathbb{M}_f)  \mathbb{M}_b^{-1} \, \pi_b,
\end{align*}
where we have used the fact that the fluid is irrotational, so that $\Psi = 0$ in \eqref{almostredham}.
Keeping in mind that $\pi = \mathfrak{m}(\pi_b)$, or alternatively that $\pi = (\mathbb{M}_b + \mathbb{M}_f) \mathbb{M}_b^{-1} \pi_b$, we finally obtain the following expression for the reduced Hamiltonian:
\[
	H_{\mathrm{red}}(g, \pi) = \frac{1}{2} \pi^T \, (\mathbb{M}_b + \mathbb{M}_f)^{-1} \, \pi,
\]
which is precisely \eqref{redhamil}.
\end{proof}

\paragraph{Summary.}  In this section, we used the framework of cotangent bundle reduction to obtain expressions for the reduced phase space, its symplectic structure, and Hamiltonian after reducing with respect to the particle relabeling symmetry.  We summarize these results in the theorem below.

\begin{theorem} \label{thm:reddiff}
	Let $\Gamma \in \mathbb{R}$  and consider the associated element $(0, \Gamma) \in \mathfrak{X}^\ast_{\mathrm{vol}}(\mathcal{F}_0)$.  Then the following properties hold:
	\begin{itemize}
		\item The isotropy group $\mathrm{Diff}_{\mathrm{vol}}(\mathcal{F}_0)_{(0, \Gamma)}$ of $(0, \Gamma)$ is the whole of the group $\mathrm{Diff}_{\mathrm{vol}}(\mathcal{F}_0)$.
		
		\item The reduced phase space $J^{-1}(0, \Gamma)/\mathrm{Diff}_{\mathrm{vol}}(\mathcal{F}_0)$ is symplectically isomorphic to $T^\ast SE(2)$, equipped with the left-invariant symplectic form given at the identity by 
		\begin{equation} \label{sympstruct}
			\Omega_\Gamma(e) := \omega_{\mathrm{can}}(e) - \Gamma \mathbf{e}^\ast_x \wedge \mathbf{e}^\ast_y, 
		\end{equation}
		where $\omega_{\mathrm{can}}$ is the canonical symplectic form on $T^\ast SE(2)$.
		
		\item The reduced kinetic energy Hamiltonian on $T^\ast SE(2)$ is given by 
		\begin{equation} \label{redhamilt}
			H_{\mathrm{red}}(g, \pi) = \frac{1}{2} \pi^T \, \mathbb{M}^{-1} \, \pi,
		\end{equation}
		where $\mathbb{M} = \mathbb{M}_b + \mathbb{M}_f$ is the full mass matrix. 
	\end{itemize}
\end{theorem}


\subsection{Reduction by the Euclidian Symmetry Group} \label{sec:eucsymm}

In the previous section, we have eliminated the particle relabeling symmetry by restricting to the fluid configurations that had circulation $\Gamma$ and no external vorticity.   However, after reducing by $\mathrm{Diff}_{\mathrm{vol}}(\mathcal{F}_0)$, the solid-fluid system is still invariant under global translations and rotations.  In other words, the group $SE(2)$ acts as a symmetry group for the reduced equations.  This is clear from Theorem~\ref{thm:reddiff}: the group $SE(2)$ acts on $T^\ast SE(2)$ by left translations, and both the reduced symplectic form as well as the Hamiltonian are invariant under that action.

\paragraph{The Reduced Poisson Structure.}  The symplectic structure \eqref{sympstruct} on $T^\ast SE(2)$ is invariant under the left action on $SE(2)$ on itself, and hence induces a Poisson bracket on the dual space $\mathfrak{se}(2)^\ast$.  If the original symplectic structure were the canonical one, the reduced Poisson bracket would simply be the (minus) Lie-Poisson bracket \eqref{lpbracket}.  However, because of the magnetic term in the symplectic structure, additional terms arise in the expression for the reduced Poisson bracket.  The general form of the magnetic Poisson bracket is described in Theorem~\ref{thm:poisson}. In the case of the rigid body with circulation, the first term in (\ref{bracket}) is the Lie-Poisson
bracket on $\mathfrak{se}(2)^\ast$, given by (\ref{liepoisson}).  The second term in (\ref{bracket}) is due to the magnetic two-form.
The entire Poisson bracket is then given by 
\begin{equation} \label{eucpoisson}
\{F, G\}_\mathcal{B} = \{F, G\}_{\mathfrak{se}(2)^\ast} - \Gamma \left(
  \frac{\partial F}{\partial P_x} \frac{\partial G}{\partial P_y} -
  \frac{\partial F}{\partial P_y} \frac{\partial G}{\partial P_x} \right).
\end{equation}
To make this bracket more explicit,
we may evaluate the Poisson bracket on the coordinate functions on $\mathfrak{se}(2)^\ast$:
\begin{equation} \label{pbcoord}
\{P_x, P_y\}_\mathcal{B} = -\Gamma, 
\quad 
\{\Pi, P_x\}_\mathcal{B} = -P_y,
\quad 
\{\Pi, P_y\}_\mathcal{B} = P_x.
\end{equation}
This emphasizes another advantage of deriving the Poisson brackets through reduction: any reduced Poisson bracket  automatically satisfies the Jacobi identity, obviating the need for any explicit computations.  While these computations would have been straightforward in this case, this is not always so (see \cite{MarsdenRatiu} for examples).

\paragraph{The Equations of Motion.}

Having established the expression for the reduced Poisson bracket on $\mathfrak{se}(2)^\ast$, we now turn to the reduced Hamiltonian.  Since $H_{\mathrm{red}}$ on $T^\ast SE(2)$ in \eqref{redhamilt} is written in terms of a left trivialization of $T^\ast SE(2)$, we immediately obtain that the reduced Hamiltonian on $\mathfrak{se}(2)^\ast$ is given by 
\begin{equation} \label{redham}
	H(\pi) = \frac{1}{2} \pi^T \, \mathbb{M}^{-1} \, \pi.
\end{equation}

The equations of motion relative to the Poisson bracket \eqref{eucpoisson} hence take exactly the form
given in~\eqref{EoM1} which we repeat here for completeness 
\begin{equation} \label{EoM}
\begin{split}
	\dot{\Pi} & = (\mathbf{P} \times \mathbf{V}) \cdot \mathbf{b}_3 \\
	\dot{\mathbf{P}} & =  \mathbf{P} \times \Omega\mathbf{b}_3 + \Gamma \mathbf{b}_3 \times \mathbf{V}
\end{split}
\end{equation}
where we have used the fact that 
\[
	\frac{\partial H}{\partial \Pi} = \Omega \quad \text{and} \quad 
		\frac{\partial H}{\partial \mathbf{P}} = \mathbf{V}.
\]
As mentioned in the introduction, these equations occurred first in the work of Chaplygin and Lamb, and were given a sound Hamiltonian foundation by \cite{BoMa06} and \cite{BoKoMa2007}.  Moreover, the equations \eqref{EoM} are a special case of the equations of motion derived by \cite{Sh2005}, \cite{BoMaRa2007} and \cite{KaOs08} for a cylinder of arbitrary shape with circulation interacting with point vortices.  
Our equations are slightly different from the ones derived by \cite{Ch1933} and used in \cite{BoMa06}, but can be brought into that form by diagonalizing the mass matrix $\mathbb{M}$ and rewriting the equations \eqref{EoM} in terms of velocities rather than the momenta.

\paragraph{The Kutta-Zhukowski Force.}  It is worthwhile to point out the geometric significance of the equations of motion.  For $\Gamma = 0$, the equations \eqref{EoM} reduce to the classical Kirchhoff equations, which are seen to be Lie-Poisson equations on $\mathfrak{se}(2)^\ast$ (a more general version of this result already appears in \cite{Leonard1997}).  When the circulation $\Gamma$ is non-zero, an additional gyroscopic force appears in the equations of motion, which is traditionally referred to as the \emph{\textbf{Kutta-Zhukowski force}}.  In coordinates, this force is given by 
\[
	 \Gamma \mathbf{b}_3 \times \mathbf{V} = \Gamma (-V_y, V_x, 0),
\]
and it follows that the force is proportional to $\Gamma$ and at right angles to $\mathbf{V}$.   From a geometric point of view, the Kutta-Zhukowski force arises because of the magnetic terms in the Poisson bracket \eqref{eucpoisson}.  Since the magnetic terms are in turn generated by the curvature of the Neumann connection, we have hence established the Kutta-Zhukowski force as a curvature-related effect.

A similar effect arises in the dynamics of a charged particle in a magnetic field.  Here, the particle moves under the influence of the Lorentz force, which is again a gyroscopic force whose magnitude is proportional to the charge $e$ of the particle.  In the Kaluza-Klein approach however, the Lorentz force becomes part of the geometry when the magnetic potential is interpreted as a connection whose curvature is precisely the magnetic field strength tensor.

\paragraph{The $\mathcal{B}_\Gamma$-potential.}  We now turn to the computation of the $\mathcal{B}\mathfrak{h}$-potential $\psi$ as in definition~\ref{def:bgpot}.  For the rigid body with circulation, we will refer to $\psi : SE(2) \rightarrow \mathfrak{se}(2)^\ast$ as the $\mathcal{B}_\Gamma$-potential, in order to emphasize the underlying magnetic form $\mathcal{B}_\Gamma$.

\begin{proposition}
	The $\mathcal{B}_\Gamma$-potential $\psi : SE(2) \rightarrow \mathfrak{se}(2)^\ast$ is given by 
	\begin{equation} \label{psiexpr}
		\psi(R_\theta, \mathbf{x}_0) =  
			\left(-\frac{\Gamma}{4} \left\Vert \mathbf{x}_0 \right\Vert^2, 
				 \frac{\Gamma}{2} \mathbb{J} \mathbf{x}_0 
			\right).
	\end{equation}
\end{proposition}
\begin{proof}
By \eqref{defbgpot}, we have that $\mathcal{B}_\Gamma(\xi_{SE(2)}, \eta_{SE(2)}) = \mathbf{d} \psi_\xi \cdot  \eta_{SE(2)}$ for $\xi, \eta \in \mathfrak{se}(2)$.  The infinitesimal generators for the left action of $SE(2)$ on itself are defined by $\xi_{SE(2)}(g) := \xi \cdot g$, and similarly for $\eta_{SE(2)}$.  Because of the left $SE(2)$-invariance of $\mathcal{B}_\Gamma$, we have that this is equivalent to
\begin{equation} \label{bgpot}
	\mathcal{B}_\Gamma(e) ( \mathrm{Ad}_{g^{-1}} \xi, \mathrm{Ad}_{g^{-1}} \eta )
		= \left< \mathbf{d} \psi_\xi , \eta_{SE(2)}(g) \right>.
\end{equation}

We now write $g = (R_\theta, \mathbf{x}_0)$, $\xi := (\Omega, \mathbf{V})$ and $\eta = (\bar{\Omega}, \bar{\mathbf{V}})$.  The vector $\mathrm{Ad}_{g^{-1}} \xi$ is then given by 
\[
\mathrm{Ad}_{(R_\theta, \mathbf{x}_0)^{-1}} (\Omega, \mathbf{V}) = (\Omega, R_\theta^T ( \mathbf{V} + \Omega \mathbf{b}_3 \times \mathbf{x}_0 )),
\]
while the infinitesimal generator $\eta_{SE(2)}$ is given by 
\[
	(\bar{\Omega}, \bar{\mathbf{V}})_{SE(2)}(R_\theta, \mathbf{x}_0) = 
	\bar{\Omega} \frac{\partial}{\partial \theta} 
		+ \left( \bar{\mathbf{V}} - \bar{\Omega} \mathbb{J} \mathbf{x}_0 \right) \cdot \nabla.
\]
Substituting these expressions into the definition \eqref{bgpot} then yields 
\[
		\psi_{(\Omega, \mathbf{V})}(R_\theta, \mathbf{x}_0) =  
			-\Omega \frac{\Gamma}{4} \left\Vert \mathbf{x}_0 \right\Vert^2
				+ \frac{\Gamma}{2} \mathbf{V}^T \mathbb{J} \mathbf{x}_0,
\]
(up to an arbitrary constant, which we set to zero)
which is equivalent to \eqref{psiexpr}.
\end{proof}

\paragraph{The Symplectic Leaves in $\mathfrak{se}(2)^\ast$.}

Proposition~\ref{prop:afforbit} offers a simple prescription for the symplectic leaves of the magnetic Poisson structure \eqref{eucpoisson}.  Using the expression \eqref{psiexpr} for the $\mathcal{B}_\Gamma$-potential, the affine action of $SE(2)$ on $\mathfrak{se}(2)^\ast$ is given by 
\begin{equation} \label{affactioncirc}
	(R_\theta, \mathbf{x}_0) \cdot (\Pi, \mathbf{P}) = 
		\left( \Pi - \mathbf{P}^T R_\theta^T \mathbb{J} \mathbf{x}_0 - 
			\frac{\Gamma}{4} \left\Vert \mathbf{x}_0 \right\Vert^2, 
			R_\theta \mathbf{P} + \frac{\Gamma}{2} \mathbb{J} \mathbf{x}_0 \right),
\end{equation}
where $(R_\theta, \mathbf{x}_0) \in SE(2)$ and $(\Pi, \mathbf{P}) \in \mathfrak{se}(2)^\ast$.  The orbits of this action are paraboloids of revolution with symmetry axis the $\Pi$-axis (see figure~\ref{fig:orbits}): they are level sets of the Casimir function 
\begin{equation} \label{orbits}
	\Phi = \Pi + \frac{1}{\Gamma} \left\Vert \mathbf{P} \right\Vert^2.
\end{equation}
Note how the case with zero circulation is a singular limit of \eqref{orbits}, and compare this with the picture of the standard coadjoint orbits in $\mathfrak{se}(2)^\ast$ (which can be obtained by setting $\Gamma = 0$ in \eqref{affactioncirc}), which are cylinders around the $\Pi$-axis, together with the individual points of the $\Pi$-axis.  The trajectories of the system lie in the intersection between the symplectic leaves and the level surfaces of the Hamiltonian; see \cite{BoMa06}.  An explicit integration by quadratures was obtained by \cite{Ch1933}.

\begin{figure}
\begin{center}
\includegraphics[scale=.5]{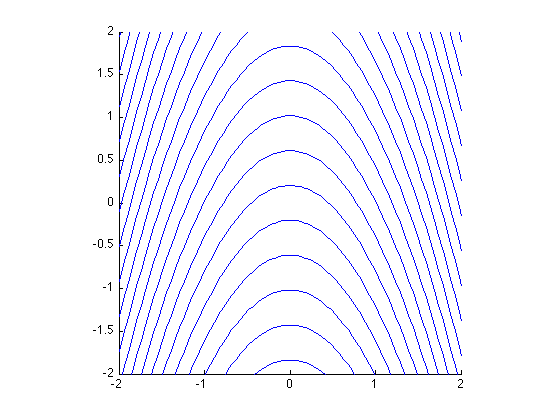}
\end{center}
\caption{Intersection with the $(\Pi, P_x)$-plane of the orbits of the affine action \eqref{affactioncirc} in $\mathfrak{se}(2)^\ast$.} \label{fig:orbits}
\end{figure}

\section{Geodesic Flow on the Oscillator Group}  
\label{sec:geodosc}

The dynamics of a rigid body with circulation bears a strong resemblance to the motion of a charged particle in a magnetic field.  In both cases, the Hamiltonian is the kinetic energy and the system is acted upon by a gyroscopic force, either the Kuta-Zhukowksi force, or the Lorentz force.  In the case of a magnetic particle, the Lorentz force can be made geometric by means of the \emph{\textbf{Kaluza-Klein}} description.   The trajectory of a magnetic particle is then a geodesic in a high-dimensional space which is the product of the original configuration space $M$ and the group $U(1)$.  See \cite{MarsdenRatiu} for an overview and \cite{Sternberg1977} for the extension to Yang-Mills fields.

One can now ask the question whether a similar description exists for the rigid body with circulation, such that its motion is geodesic in an appropriate higher-dimensional space.  Surprisingly, the answer turns out to be positive.  The relevant space is the \textbf{\emph{oscillator group}} $\Osc$, which is a central extension of $SE(2)$ by the real line $\mathbb{R}$.  Below, we analyze the structure of this group and we give an explicit expression for the Lie-Poisson bracket on the dual of its Lie algebra.  We then show that the equations of motion induced by this bracket are closely related to the equations of motion for the rigid body with circulation.

The use of central extensions in the study of mechanical systems is not new.  We mention here in particular the work of \cite{OvKh1987} on the KdV equation as a geodesic equation on the Virasoro group, and the work of \cite{Vi2001}, in which the structure of the geodesic equations on an extension of a Lie group is studied in general.  A detailed account of the geometry of central extensions, including the oscillator group, can be found in \cite{MarsdenHamRed}.

\paragraph{The Oscillator Group.}
 The oscillator group $\Osc$ is a central extension of $SE(2)$, \emph{i.e.} there is an exact sequence 
\[
	0 \rightarrow \mathbb{R} \rightarrow \Osc \rightarrow SE(2) \rightarrow \{e\}.
\]
The multiplication in $\Osc$ is determined by the specification of a real-valued $SE(2)$-two-cocycle $B : SE(2) \times SE(2) \rightarrow \mathbb{R}$:
\begin{equation} \label{groupcocycle}
	B( (R_\theta, \mathbf{x}_0), (R_\psi, \mathbf{y}_0) ) := \omega_{\mathbb{R}^2}(\mathbf{x}_0, R_\theta \mathbf{y}_0) = \mathbf{x}_0 \cdot \mathbb{J} R_\theta \mathbf{x}_0.
\end{equation}
Here, $\omega_{\mathbb{R}^2}$ is the standard symplectic area form on $\mathbb{R}^2$:
\[
	\omega_{\mathbb{R}^2}(\mathbf{x}, \mathbf{y}) := \mathbf{x}\cdot \mathbb{J} \mathbf{y},
\]
where $\mathbb{J}$ is the symplectic matrix \eqref{sympmat}.

The Lie algebra of the oscillator group is denoted as $\osc$ and has as its underlying vector space the space $\mathfrak{se}(2) \times \mathbb{R}$.  The Lie bracket is determined by the Lie algebra two-cocycle $C : \mathfrak{se}(2) \times \mathfrak{se}(2) \rightarrow \mathbb{R}$ induced by $B$ as follows
\[
	C(v_1, v_2) := \frac{\partial^2}{\partial s \partial t} \Big|_{s, t = 0} ( B(g(t), h(s)) - B(h(s), g(t)) ),
\]
where $g(t)$ and $h(s)$ are smooth curves in $SE(2)$ such that $\dot{g}(0) = v_1$ and $\dot{h}(0) = v_2$.  Explicitly, $C$ is given by 
\begin{equation} \label{algcocycle}
	C((\Omega_1, \mathbf{V}_1), (\Omega_2, \mathbf{V}_2)) = 2 \omega_{\mathbb{R}^2}(\mathbf{V}_1, \mathbf{V}_2) = 2 \mathbf{V}_1 \cdot \mathbb{J} \mathbf{V}_2.
\end{equation}
The bracket of the oscillator algebra is then given by 
\begin{align*}
	[(\Omega_1, \mathbf{V}_1, a), (\Omega_2, \mathbf{V}_2, b)] & =  ([(\Omega_1, \mathbf{V}_1), (\Omega_2, \mathbf{V}_2)], 
		C((\Omega_1, \mathbf{V}_1), (\Omega_2, \mathbf{V}_2)) \\
		& = (0, -\Omega_1 \mathbb{J}\mathbf{V}_2 + \Omega_2 \mathbb{J}\mathbf{V}_1, 
			2 \mathbf{V}_1 \cdot \mathbb{J} \mathbf{V}_2).
\end{align*}

Further information about the oscillator group and its algebra can be found in \cite{Streater1967} and 
\cite{MarsdenHamRed}.

\paragraph{The Lie-Poisson Bracket on $\osc^\ast$.}
We now determine the Lie-Poisson bracket on the dual Lie algebra $\osc^\ast$ and relate this expression with the Poisson bracket for the rigid body with circulation.
Note first that as a vector space, $\osc^\ast$ is just $\mathfrak{se}(2)^\ast \times \mathbb{R}$, so that an element $\nu$ of $\osc^\ast$ can be written as $\nu := (\pi, p )$, where $\pi = (\Pi, \mathbf{P}) \in \mathfrak{se}(2)^\ast$ and $p \in \mathbb{R}$.

\begin{proposition} \label{prop:osc}
	The Lie-Poisson bracket on $\osc^\ast$ is given by 
	\begin{equation} \label{oscbracket}
		\{ f, g \}_{\osc^\ast}(\pi, p) = 
		\{ f, g \}_{\mathfrak{se}(2)^\ast}(\pi) - p C\left(\frac{\delta f}{\delta \pi}, \frac{\delta g}{\delta \pi} \right),
	\end{equation}
	where $C$ is the $\mathfrak{se}(2)$-two-cocycle given by (\ref{algcocycle}).
\end{proposition}

\begin{proof}
As usual, it is sufficient to determine the Lie-Poisson bracket on linear functions $f, g$ on $\osc^\ast$, for which
\[
	f(\nu) = \left\langle \nu, \frac{\delta f}{\delta \nu} \right\rangle, \quad \text{for $\frac{\delta f}{\delta \nu} \in \osc$},
\]
and similar for $g$.  Since $\nu \in \osc^\ast$, we may write $\nu = (\pi, p)$, where $\pi \in \mathfrak{se}(2)^\ast$ and $p \in \mathbb{R}$.  Likewise, we may decompose the variational derivative as 
\[
	\frac{\delta f}{\delta \nu} = \left(\frac{\delta f}{\delta \pi}, \frac{\delta f}{\delta p} \right) 
		\in \mathfrak{se}(2) \times \mathbb{R}.
\] 
The minus Lie-Poisson bracket on $\osc^\ast$ is then given by a similar formula as \eqref{lpbracket}, and we get
\begin{align*}
	\{f, g\}_{\osc^\ast} & = -\left<\nu, \left[\frac{\delta f}{\delta \nu}, \frac{\delta g}{\delta \nu} \right] \right> \\
	& = -\left< (\pi, p), \left( \left[\frac{\delta f}{\delta \pi}, \frac{\delta g}{\delta \pi} \right], C\left(\frac{\delta f}{\delta \pi}, \frac{\delta g}{\delta \pi} \right) \right) \right> \\
	& = -\left< \pi, \left[\frac{\delta f}{\delta \pi}, \frac{\delta g}{\delta \pi} \right] \right>
		- p C\left(\frac{\delta f}{\delta \pi}, \frac{\delta g}{\delta \pi} \right),
\end{align*}
so that we obtain the expression (\ref{oscbracket}).
\end{proof}

In coordinates, the Poisson structure of proposition~\ref{prop:osc} is given by 
\begin{equation} \label{poissosc}
\{ f, g \}_{\osc^\ast}(\pi, p) = 
	(\nabla_{(\pi, p)} f)^T
		\begin{pmatrix} 
    0 & -P_y & P_x & 0 \\
    P_y & 0 & -p & 0 \\
    -P_x & p & 0 &0 \\
    0 & 0 & 0 & 0
    \end{pmatrix}
    	(\nabla_{(\pi, p)} g).
\end{equation}

These coordinate expressions also serve to clarify the link between the description of the rigid body on the oscillator group and the Euclidian group: each of the subspaces $\mathfrak{se}(2)^\ast$ is a Poisson submanifold of $\mathrm{osc}^\ast$, as in the following proposition.

\begin{proposition} \label{prop:inc}
	Let $\Gamma \in \mathbb{R}$.  The inclusion $\iota_\Gamma : \mathfrak{se}(2)^\ast \hookrightarrow \mathrm{osc}^\ast$ given by $\iota_\Gamma(\pi) = (\pi, \Gamma)$ is a Poisson map.
\end{proposition}
\begin{proof}
This can easily be seen from the coordinate expressions above: for $p = \Gamma$, the Poisson structure \eqref{oscbracket} coincides with \eqref{eucpoisson} on the submanifold $\mathfrak{se}(2)^\ast$.
\end{proof}

\paragraph{The Equations of Motion.}  The reduced kinetic energy Hamiltonian \eqref{redham} on $\mathfrak{se}(2)^\ast$ gives rise to a Hamiltonian ${H}_{\osc^\ast}$ on $\osc^\ast$ as follows:
\begin{align*}
	{H}_{\osc^\ast}(\pi, p) & = H(\pi) + \frac{p^2}{2}  = \frac{1}{2} \pi^T \, \mathbb{M}^{-1} \, \pi + \frac{p^2}{2}.
\end{align*}
The equations of motion for this Hamiltonian and the bracket \eqref{oscbracket} are then given in coordinates by 
\begin{align}  \label{eomosc}
	\dot{\Pi} & = (\mathbf{P} \times \mathbf{V}) \cdot \mathbf{b}_3, \nonumber \\
	\dot{\mathbf{P}} & =  \mathbf{P} \times \Omega\mathbf{b}_3 + p \mathbf{b}_3 \times \mathbf{V}, \\
	\dot{p} & = 0. \nonumber
\end{align}

\paragraph{Conclusions.}  We summarize the conclusions of this section in the following theorem.

\begin{theorem} Let $\Gamma$ be an element of $\mathbb{R}^2$.
\begin{itemize} 

	\item The Lie-Poisson structure on the oscillator algebra is given by \eqref{oscbracket} and the resulting equations of motion by \eqref{eomosc}.

	\item The dynamics of a rigid body with circulation $\Gamma$ takes place on the Poisson submanifold $p = \Gamma$ of $\mathrm{osc}^\ast$.

\end{itemize}
\end{theorem}


\section{Conclusions and Outlook}
\label{sec:outlook}

In this paper, we established the non-canonical Hamiltonian structure of \cite{BoMa06} by geometric means, and we showed how the geometric description gives new insight into classical results such as the Kutta-Zhukowski force.  We now discuss a number of open questions related to this description. 

\paragraph{Cocycles.}  The problem considered here seems to fit in a general class of systems where a remarkable interaction between cocycles and curvature forms is at play; examples include the work of \cite{HoKu1988} and \cite{CeMaRa2004} on the analogy between spin glasses and Yang-Mills fluids.  Recently, a comprehensive reduction theory has been developed for such systems (see \cite{GaRa2008} and \cite{GaRa2009}). It would be interesting to see whether this reduction theory can be directly applied to the rigid body with circulation considered here.

\paragraph{The Oscillator Group and Reduction.}  It is still not entirely clear why the description in terms of the oscillator group, highlighted in proposition~\eqref{prop:inc}, should exist.  Is there some way of deriving the oscillator group immediately from the unreduced space $Q$ and the particle relabeling symmetry?  
One way to do this goes back to the work of Kelvin, who dealt with circulation by placing a fixed surface in the fluid so that the fluid domain becomes simply connected.  The flow across the surface then appears as a new variable, and the reduced variables are precisely the coordinates on $\mathrm{Osc}$.  Kelvin's argument originally only held for fluids in a bounded container, but can be extended to unbounded flows as well.  This approach will be the subject of a future work.

%

%

\paragraph{Point Vortices and Vortical Structures.}  In \cite{VaKaMa2009}, the dynamics of a rigid body interacting with point vortices was investigated from a geometric point of view.  
It would be of interest to extend the results of this paper to the case of a rigid body with circulation moving in a field of point vortices or a general distribution of vorticity and, in particular, it is interesting to investigate whether one could
generalize the oscillator group description to the case of non-zero vorticity.
Lastly, we also intend to study the geometry of controlled articulated bodies moving in the field of point vortices or other vortical structures.  This setup has important applications in the theory of locomotion; see \cite{KMlocomotion} and~\cite{Kanso2009}.


\appendix

\section{The Rigid Body Group}
\label{appendix:rigidgroup}

In this appendix we discuss some elementary aspects of the special Euclidian group $SE(2)$, consisting of translations and rotations.  This material is standard and can be found, for instance, in \cite{MarsdenRatiu} and \cite{ArKoNe1997}.

\paragraph{The Rigid Body Group $SE(2)$.} An element $g \in SE(2)$ consists of a rigid rotation and a rigid translation: $g \equiv (R_\theta, \mathbf{x}_0)$, where $R_\theta \in SO(2)$ and $\mathbf{x}_0 \in \mathbb{R}^2$. 
The group multiplication and inversion in $SE(2)$ are respectively given by 
\[
	(R_\theta, \mathbf{x}_0) \cdot (R_\psi, \mathbf{y}_0) = 
	(R_{\theta + \psi}, \mathbf{x}_0 + R_\theta \mathbf{y}_0)
	\quad \text{and} \quad 
	(R_\theta, \mathbf{x}_0)^{-1} = (R_{-\theta}, - R_{-\theta} \mathbf{x}_0).
\]
It will often be useful to write an element $(R_\theta, \mathbf{x}_0) \in SE(2)$ as a matrix 
\begin{equation}  \label{matse}
 (R_\theta, \mathbf{x}_0) \leadsto
  \begin{pmatrix}
    R_\theta & \mathbf{x}_0 \\
    0 & 1
  \end{pmatrix}.
\end{equation}
Under this identification, 
the group composition and inversion in $\operatorname{SE}(2) $
are given by matrix multiplication and inversion.  The Lie algebra $\mathfrak{se}(2)$ of
$\operatorname{SE}(2)$ is parametrized by $\zeta =  (\Omega, \mathbf{V})^T$ and 
can be identified with $\mathbb{R}^3$ with the following Lie bracket:
\[
	[(\Omega_1, \mathbf{V}_1), (\Omega_2, \mathbf{V}_2)] 
		= (0, -\Omega_1 \mathbb{J} \mathbf{V}_2 + \Omega_2 \mathbb{J} \mathbf{V}_2)
\]
for all $\Omega_{1,2} \in \mathbb{R}$ and $\mathbf{V}_{1,2} \in \mathbb{R}^2$, where 
$\mathbb{J}$ is the symplectic matrix given by 
\begin{equation} \label{sympmat}
	\mathbb{J} = 
		\begin{pmatrix}
			0 & 1 \\
			-1& 0 
		\end{pmatrix}.
\end{equation}
We will often use the following basis of $\mathfrak{se}(2)$:
\begin{equation} \label{sebasis}
	\mathbf{e}_\Omega = (1, 0, 0)^T, \quad 
	\mathbf{e}_x = (0, 1, 0)^T \quad
	\mathbf{e}_y = (0, 0, 1)^T.
\end{equation}
The dual $\mathfrak{se}(2)^\ast$ of the Lie algebra of $SE(2)$ can also be identified with $\mathbb{R}^3$  by using the Euclidian inner product on $\mathfrak{se}(2) = \mathbb{R}^3$.  We denote a typical element of $\mathfrak{se}(2)^\ast$ by $\pi = (\Pi, \mathbf{P})$ with $\Pi \in \mathbb{R}$ and $\mathbf{P} = (P_x, P_y) \in \mathbb{R}^2$.  The notation is chosen to be reminiscent of angular and linear momentum in the body frame.  The Lie-Poisson bracket%
\footnote{We use the \emph{minus} Lie-Poisson bracket, which is the relevant choice of Poisson bracket for rigid body dynamics.  More information about this choice can be found in \cite{MarsdenRatiu}.} 
on $\mathfrak{se}(2)^\ast$ is then given by the standard formula:
\begin{equation} \label{lpbracket}
\{F, G\}_{\mathfrak{se}(2)^\ast}(\pi) =
-\left< \pi, \left[\frac{\delta F}{\delta \pi}, \frac{\delta G}{\delta \pi} \right] \right>
\end{equation}
for functions $F(\Pi, \mathbf{P}), G(\Pi, \mathbf{P})$ on $\mathfrak{se}(2)^\ast$, 
and since 
\[
	\frac{\delta F}{\delta \pi} := \nabla_{(\Pi, \mathbf{P})}  F= \left( \frac{\partial F}{\partial \Pi}, \nabla_{\mathbf{P}} F \right)
\]
we have that the bracket is explicitly given  by 
\begin{equation} \label{liepoisson}
\{F, G\}_{\mathfrak{se}(2)^\ast} = 
  (\nabla_{(\Pi, \mathbf{P})}  F)^T \Lambda \nabla_{(\Pi, \mathbf{P})}  G,
\quad 
\text{where}
\quad
  \Lambda = \begin{pmatrix} 
    0 & -P_y & P_x \\
    P_y & 0 & 0 \\
    -P_x & 0 & 0 
    \end{pmatrix}.
\end{equation}

For future reference, we denote by 
\begin{equation} \label{dualbasis}
\mathbf{e}^\ast_\Omega = (1, 0, 0), \quad 
	\mathbf{e}^\ast_x = (0, 1, 0) \quad
	\mathbf{e}^\ast_y = (0, 0, 1)
\end{equation}
the basis of $\mathfrak{se}(2)^\ast \cong \mathbb{R}^3$ which is dual to the basis (\ref{sebasis}) of $\mathfrak{se}(2)$.

\section{The Group of Volume-Preserving Diffeomorphisms} 
\label{appendix:diffgroup}

In this appendix we recall some properties of the group of volume preserving diffeomorphisms and its Lie algebra.  This material is well-known and can be found for instance in \cite{Ar66}, \cite{EbinMarsden70}, and \cite{ArnoldKhesin}.

The fluid domain is a subset of $\mathbb{R}^2$. The reference configuration of the fluid is denoted by $\mathcal{F}_0$ and the space taken by the fluid at a generic
time $t$ is denoted by $\mathcal{F}$. The fluid domain $\mathcal{F}$ is a 
submanifold of $\mathbb{R}^2$ and hence it inherits the Euclidian metric from $\mathbb{R}^2$, 
which we denote by $\left<\!\left< v, w \right>\!\right>$ for all $v, w \in T\mathcal{F}$.
The associated Hodge star operator will be denoted by $\ast : \Omega^k(\mathcal{F}) \rightarrow \Omega^{2-k}(\mathcal{F})$, and using $\ast$ we may write the induced metric on the space of forms by 
\begin{equation} \label{formnorm}
	\left<\!\left< \alpha, \beta \right>\!\right> := \int_{\mathcal{F}} \alpha \wedge \ast \beta.
\end{equation}
We also introduce the co-differential $\delta: \Omega^k(\mathcal{F}) \rightarrow \Omega^{k-1}(\mathcal{F})$ by the standard prescription $\boldsymbol{\delta} \alpha = \ast \, \mathbf{d} \ast \alpha$.
Lastly, we introduce the index raising and lowering isomorphisms induced by the metric, and write them respectively as $\sharp : T^\ast \mathcal{F} \rightarrow T\mathcal{F}$, where $\left<\!\left< \sharp(\alpha) , v \right>\!\right> = \alpha(v)$, and $\flat = \sharp^{-1}$.

An incompressible motion of the fluid is a one-parameter (time) family of elements of the group $\mathrm{Diff}_{\mathrm{vol}}(\mathcal{F}_0)$ of volume-preserving diffeomorphisms, 
which encodes the symmetry of the fluid by arbitrary particle relabelings. 
Formally speaking, $\mathrm{Diff}_{\mathrm{vol}}(\mathcal{F}_0)$ is a Lie group with as its Lie algebra the space $\mathfrak{X}_{\mathrm{vol}}(\mathcal{F}_0)$ of divergence-free vector fields which are tangent to the boundary of the body $\mathcal{B}$.  The Lie bracket is the negative of the commutator of vector fields. The dual space $\mathfrak{X}^\ast_{\mathrm{vol}}(\mathcal{F}_0)$ can be identified with the space $\Omega^1(\mathcal{F}_0)/\mathbf{d}\Omega^0(\mathcal{F}_0)$ of one-forms up to exact forms with duality pairing 
\begin{equation} \label{dual}
	\left<X, [\alpha]\right>= \int \alpha(X) \, dV,
\end{equation}
but a better identification can be made.  Any equivalence class $[\alpha] \in \mathfrak{X}^\ast_{\mathrm{vol}}(\mathcal{F}_0)$ is uniquely determined by the exterior derivative $\mathbf{d}\alpha$ and by its value on a basis of the first homology of $\mathcal{F}_0$.  Because of the presence of a single rigid body, the homology is generated by any closed curve $\mathcal{C}$ around the body, and the value of $\alpha$ on $\mathcal{C}$ is simply given by 
\begin{equation} \label{gamma}
	\Gamma := \int_\mathcal{C} \alpha.
\end{equation}
The derivative $\mathbf{d}\alpha$ corresponds to the \textbf{\emph{vorticity}} of the system, while $\Gamma$ corresponds to the \textbf{\emph{circulation}}.  We hence have an isomorphism between $\mathfrak{X}^\ast_{\mathrm{vol}}(\mathcal{F}_0)$ and $\mathbf{d}\Omega^1(\mathcal{F}_0) \times \mathbb{R}$ given by 
\begin{equation} \label{iso}
	[\alpha] \mapsto (\mathbf{d}\alpha, \Gamma),
\end{equation}
where $\Gamma$ is given by (\ref{gamma}).

We finish by noting that $\mathrm{Diff}_{\mathrm{vol}}(\mathcal{F}_0)$ acts on $\mathfrak{X}_{\mathrm{vol}}(\mathcal{F}_0)$ and $\mathfrak{X}^\ast_{\mathrm{vol}}(\mathcal{F}_0)$ through the adjoint and the
co-adjoint action, respectively.  These actions are both given by push-forward: if $\phi$ is an element
of $\mathrm{Diff}_{\mathrm{vol}}(\mathcal{F}_0)$, and $\mathbf{u}$ and $(\mathbf{d} \alpha, \Gamma)$ are elements of
$\mathfrak{X}_{\mathrm{vol}}(\mathcal{F}_0)$ and $\mathfrak{X}^\ast_{\mathrm{vol}}(\mathcal{F}_0)$, respectively, then
\begin{equation} \label{coad}
	\mathrm{Ad}_\phi(\mathbf{u}) = \phi_{\ast}  \mathbf{u}, \quad\text{and}\quad 
 \mathrm{CoAd}_\phi(\mathbf{d} \alpha, \Gamma) = (\mathbf{d} (\phi_{\ast}  \alpha), \Gamma).
\end{equation}
The fact that the co-adjoint action leaves the circulation $\Gamma$
invariant is related to Kelvin's theorem, which says that the circulation around a material loop in the fluid is constant.


\end{document}